%% file: main.tex
\documentclass[10pt,conference]{IEEEtran}
\usepackage{cite}
\usepackage{amsmath,amssymb,amsfonts}
\usepackage{algorithmic}
\usepackage{graphicx}
\usepackage{textcomp}
\PassOptionsToPackage{prologue,usenames,dvipsnames,table}{xcolor}
\usepackage{xcolor}
\usepackage[hyphens]{url}
\usepackage{fancyhdr}
\usepackage{hyperref}
\usepackage{booktabs}
\usepackage{multirow}
\usepackage{array}
\usepackage{makecell} 
\usepackage{threeparttable}
\usepackage{enumitem}
\usepackage{braket}
\usepackage[mathscr]{euscript}
\usepackage{svg} % 需使用包
\usepackage{setspace}
\usepackage{xspace}
\usepackage{soul}
\soulregister\cite7
\soulregister\citep7
\soulregister\citet7
\soulregister\ref7
\soulregister\pageref7
\usepackage{pifont}
\usepackage[ruled,vlined]{algorithm2e}
\usepackage{setspace}
\usepackage{amsthm}
\IEEEoverridecommandlockouts

\newcommand{\hpcayear}{2025}

\newcommand{\hpcasubmissionnumber}{590}

\title{\papername: \underline{C}ommute \underline{H}amiltonian-based \underline{Q}A\underline{O}A \\ for \underline{C}onstrained Binary \underline{O}ptimization}

\newcounter{KPindex}

\newcounter{GCindex}

\newcounter{FLindex}

\newcommand{\papername}{Choco-Q\xspace}

\newcommand{\blackxmark}{\textcolor{black}{\ding{55}}}%

\newtheorem{lemma}{Lemma}

\def\hpcacameraready{} % Uncomment to build camera-ready version

\newcommand\hpcaauthors{Debin Xiang\IEEEauthorrefmark{2} and Qifan Jiang\thanks{\IEEEauthorrefmark{2} These authors contributed equally to this work.}\IEEEauthorrefmark{2} and Liqiang Lu\IEEEauthorrefmark{1}\thanks{\IEEEauthorrefmark{1} Corresponding authors.} and Siwei Tan and Jianwei Yin}
\newcommand\hpcaaffiliation{Zhejiang University}
\newcommand\hpcaemail{\{db.xiang, jiang7f, liqianglu, siweitan\}@zju.edu.cn, zjuyjw@cs.zju.edu.cn}

\author{
  \ifdefined\hpcacameraready
    \IEEEauthorblockN{\hpcaauthors{}}
      \IEEEauthorblockA{
        \hpcaaffiliation{} \\
        \hpcaemail{}
      }
  \else
    \IEEEauthorblockN{\normalsize{HPCA \hpcayear{} Submission
      \textbf{\#\hpcasubmissionnumber{}}} \\
      \IEEEauthorblockA{
        Confidential Draft \\
        Do NOT Distribute!!
      }
    }
  \fi 
}
\fancypagestyle{camerareadyfirstpage}{%
  \fancyhead{}
  
  \fancyhead[C]{
    \ifdefined\aeopen
    \parbox[][12mm][t]{13.5cm}{\hpcayear{} IEEE International Symposium on High-Performance Computer Architecture (HPCA)}    
    \else
      \ifdefined\aereviewed
      \parbox[][12mm][t]{13.5cm}{\hpcayear{} IEEE International Symposium on High-Performance Computer Architecture (HPCA)}
      \else
      \ifdefined\aereproduced
      \parbox[][12mm][t]{13.5cm}{\hpcayear{} IEEE International Symposium on High-Performance Computer Architecture (HPCA)}
      \else
      \parbox[][0mm][t]{13.5cm}{\hpcayear{} IEEE International Symposium on High-Performance Computer Architecture (HPCA)}
    \fi 
    \fi 
    \fi 
    \ifdefined\aeopen 
      \includegraphics[width=12mm,height=12mm]{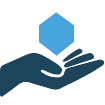}
    \fi 
    \ifdefined\aereviewed
      \includegraphics[width=12mm,height=12mm]{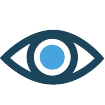}
    \fi 
    \ifdefined\aereproduced
      \includegraphics[width=12mm,height=12mm]{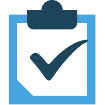}
    \fi
  }
  \fancyfoot[C]{}
}
\usepackage{forest}
\definecolor{folderbg}{RGB}{124,166,198}
\definecolor{folderborder}{RGB}{110,144,169}
\def\Size{4pt}
\tikzset{
  folder/.pic={
    \filldraw[draw=folderborder,top color=folderbg!50,bottom color=folderbg]
      (-1.05*\Size,0.2\Size+5pt) rectangle ++(.75*\Size,-0.2\Size-5pt);  
    \filldraw[draw=folderborder,top color=folderbg!50,bottom color=folderbg]
      (-1.15*\Size,-\Size) rectangle (1.15*\Size,\Size);
  }
}

\usepackage{listings}
\begin{document}
\maketitle

%Enables the camera ready header and footer
\ifdefined\hpcacameraready 
  \thispagestyle{camerareadyfirstpage}
  \pagestyle{empty}
\else
  \thispagestyle{plain}
  \pagestyle{plain}
\fi

\newcommand{\hpcaheight}{0mm}
\ifdefined\eaopen
\renewcommand{\hpcaheight}{12mm}
\fi

% \footnote{$\dagger$These authors contributed equally to this work.\\$^*$Corresponding author}
% \begin{document}
% \maketitle

\begin{abstract}
Constrained binary optimization aims to find an optimal assignment to minimize or maximize the objective meanwhile satisfying the constraints, which is a representative NP problem in various domains, including transportation, scheduling, and economy. Quantum approximate optimization algorithms (QAOA) provide a promising methodology for solving this problem by exploiting the parallelism of quantum entanglement. However, existing QAOA approaches based on penalty-term or Hamiltonian simulation fail to thoroughly encode the constraints, leading to extremely low success rate and long searching latency. 

% Constrained binary optimization is to find an optimal assignment for binary variables to minimize the objective while satisfying the constraints. It is an important non-deterministic polynomial-time complete problem that plays a major role in various fields, including transportation, engineering, scheduling, and finance. Quantum Approximate Optimization Algorithms provide a promising methodology for solving constrained binary optimization problems by exploiting the parallelism of quantum entanglement, where each variable is represented by a qubit. However, the current Quantum Approximate Optimization Algorithm and its variants cannot handle all constraints and suffer from a very low success rate and slow solving speed.

This paper proposes \papername, a formal and universal framework for constrained binary optimization problems, which comprehensively covers all constraints and exhibits high deployability for current quantum devices. The main innovation of \papername is to embed the commute Hamiltonian as the driver Hamiltonian, resulting in a much more general encoding formulation that can deal with arbitrary linear constraints. Leveraging the arithmetic features of commute Hamiltonian, we propose three optimization techniques to squeeze the overall circuit complexity, including Hamiltonian serialization, equivalent decomposition, and variable elimination. The serialization mechanism transforms the original Hamiltonian into smaller ones. Our decomposition methods only take linear time complexity, achieving end-to-end acceleration. Experiments demonstrate that \papername shows more than 235$\times$ algorithmic improvement in successfully finding the optimal solution, and achieves 4.69$\times$ end-to-end acceleration, compared to prior QAOA designs.

% a high-confident and universal quantum solver for constrained binary optimization problems. \papername is the first approach that embeds the commute Hamiltonian into Quantum Approximate Optimization Algorithms to deploy the Hamiltonian into quantum hardware, we propose a serialization of commute Hamiltonian to replace the original complex Hamiltonian and a phase gate decomposition to deploy the serialized Hamiltonian into quantum hardware. Finally, we evaluate the importance of variables on the hardware and eliminate the key variables to improve scalability under real-world quantum hardware. Our experiments demonstrate that \papername increases the success rate by 235$\times$ on average on a noise-free simulator and achieves up to 4.69$\times$ end-to-end acceleration compared with state-of-the-art works.

\end{abstract}
% \keywords{Quantum Computing, Quantum Annealing, SAT problem}

\input{section/1_introduction}

\input{section/2_background}

\input{section/3_formulation}
\input{section/4_opt2}
\input{section/6_evaluation}
\input{section/8_related_work}

\input{section/9_conclusion}

\input{section/10_acknowledgement.tex}
\input{section/appendix}
% \section*{Acknowledgements}
% This document is derived from previous conferences, in particular ISCA 2020.

%%%%%%% -- PAPER CONTENT ENDS -- %%%%%%%%

%%%%%%%%% -- BIB STYLE AND FILE -- %%%%%%%%
\bibliographystyle{IEEEtranS}
\bibliography{reference}
%%%%%%%%%%%%%%%%%%%%%%%%%%%%%%%%%%%%

\end{document}

%% file: section/1_introduction.tex
\section{Introduction}

The constrained binary optimization is to find an assignment of binary variables that minimize or maximize the objective function under the constraints. Typically, the constraints usually comprise multiple linear equations, namely equality, e.g., $C\vec{u}=\vec{c}$. There are various representative applications of this problem, such as facility location~\cite{facilityoverview}, project scheduling~\cite{projectscheduling}, portfolio optimization~\cite{portfoliooptimization}, political districting~\cite{politicaldistricting}, and energy system optimization~\cite{energyopt}. The constrained binary optimization has been demonstrated as an NP-hard problem that takes exponential time and space complexity for the classical solver to get the exact solution~\cite{karp2010reducibility,gurobi}.

\input{tab/relatedwork}

Quantum computing stands as a revolutionary paradigm to handle this problem by taking advantage of quantum entanglement and superposition. Notably, the quantum approximate optimization algorithm (QAOA) is a class of variational algorithms, which consists of objective Hamiltonian and driver Hamiltonian to iteratively search the optimal solution~\cite{qaoaorigin,verma2022penalty,redqaoa,qaoacompilation,frozenqubits,qaoacompile}. Mathematically, it encodes the objective function and constraints into Hamiltonians and updates the parameters during simulation. In particular, the penalty-based method is a natural idea that inserts the constraint into the objective function as penalty terms~\cite{ayodele2022penalty,brandhofer2022benchmarking,verma2022penalty}. While, another approach is to encode the constraints into the driver Hamiltonian~\cite{cyclicdriver,hadfield2017quantum,quantum_alternating_operator_ansatz}.

However, existing algorithms are severely limited by the low success rate, originating from the inability to thoroughly encode the constraints. Table~\ref{tab:comparisionwork} summarizes several features of previous QAOA designs and compares them on graph coloring problem~\cite{graphcoloring} using a 15-qubit simulator.
Here, we focus on two critical metrics: 1) \textit{in-constraints rate}, the probability that the output solutions satisfy all constraints; 2) \textit{success rate}, the probability of getting the optimal solution. We observe an extremely low success rate across all prior methods. Especially, the penalty-based methods even exhibit a close-to-zero success rate, smashing quantum advantages. 

Fundamentally, the penalty-based method only supports soft constraint encoding, which highly depends on the coefficient of the penalty terms. Concretely, a small penalty coefficient may fail to take the constraints into account, as shown in Figure~\ref{fig:intro}~(a). Nevertheless, a high coefficient will diminish the gap between the optimal solution and non-optimal ones when calculating the objective, decreasing the success rate. Thus, tuning the penalty coefficient gives rise to lots of time involved in classical computing. For example, in Table~\ref{tab:comparisionwork}, Red-QAOA~\cite{redqaoa} takes 16.7s, composed of the classical part and quantum part, and shows an inferior success rate of 0.03\%.

Note that the penalty-based QAOA inherently expresses the Hamiltonian matrix described by the Ising model~\cite{isingmodel}, which is a mathematical representation of the evolution of a quantum system. Therefore, an alternative approach to encode constraints is to develop other types of Hamiltonian, such as the quantum alternating operator ansatz that encodes the constraints into the driver Hamiltonian~\cite{quantum_alternating_operator_ansatz,cyclicdriver}. However, prior Hamiltonian-based method has to restrict the constraints in a specific format, lacking the generality to address arbitrary linear equality. For example, the cyclic Hamiltonian-based simulation~\cite{cyclicdriver} only supports the equality described in summation format (e.g., $x_1+x_2+x_3=2$, or $-x_1-x_2-x_3=-1$). Thus, when coping with complex constraints, it may locate solutions in the non-constrained space, as shown in Figure~\ref{fig:intro}~(a).
These approaches also exhibit poor success rates, as shown in Table~\ref{tab:comparisionwork}. Moreover, the implementation of Hamiltonian requires decomposing it into deployable basic gates, accounting for exponential time complexity, leading to long end-to-end latency.

\begin{figure}[t]
    \centering
    \includegraphics[width=0.98\linewidth]{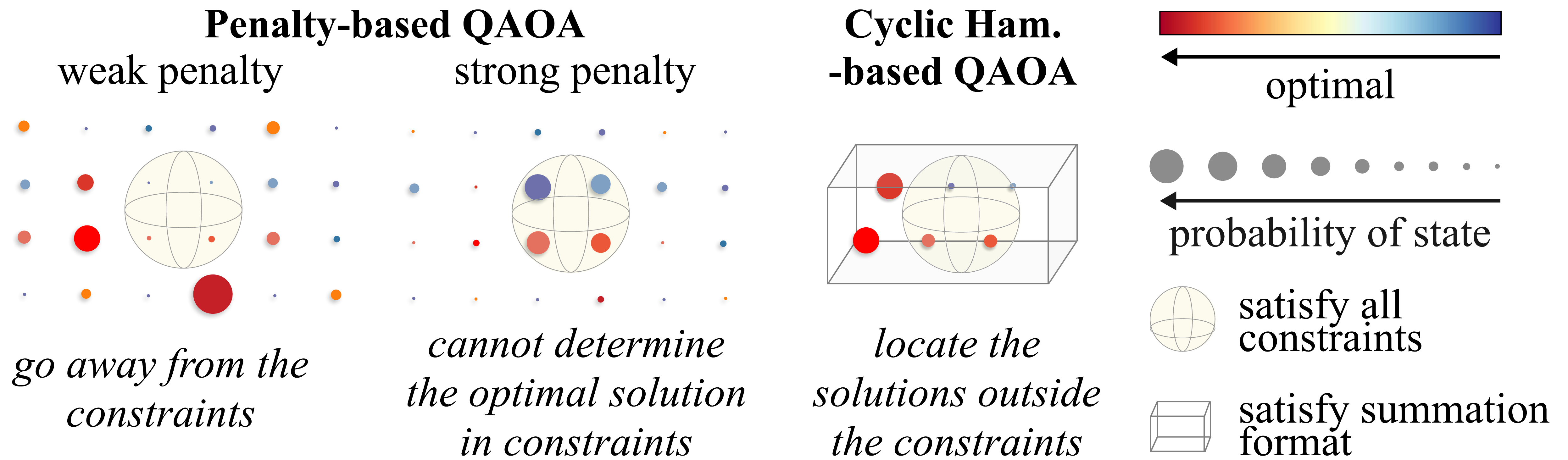}
     \\
     \vspace{-0.1cm}
    \makebox[0.90\linewidth]{\footnotesize (a) The solution space under prior QAOA designs.}
    \vspace{0.25cm} 
    \\
    \includegraphics[width=.98\linewidth]{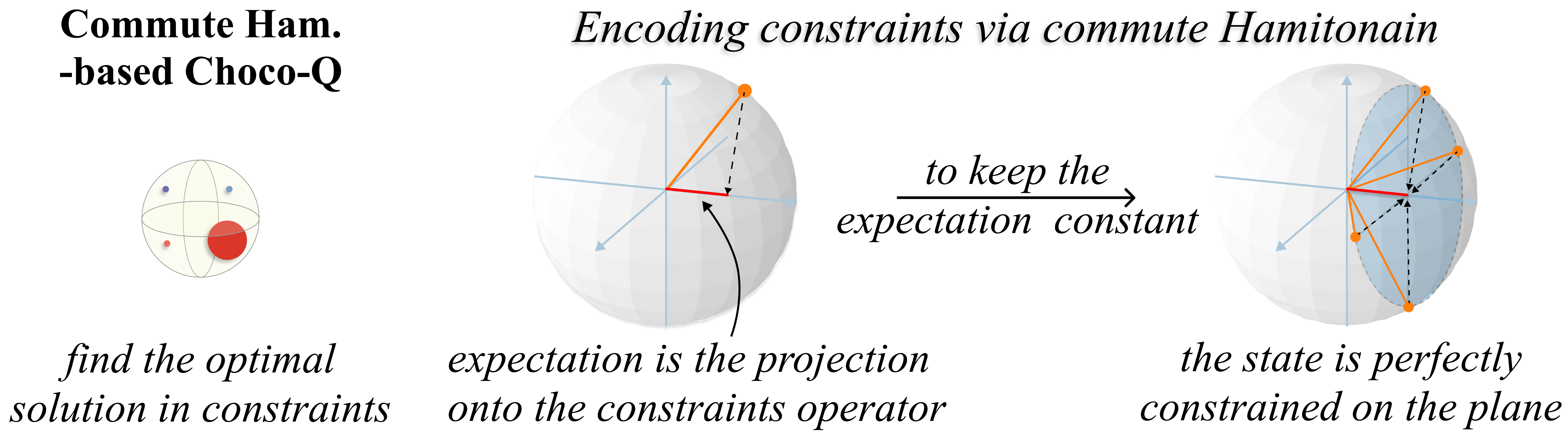}
    \\
    \vspace{-0.1cm} 
    \makebox[0.90\linewidth]{\footnotesize (b) The foundation of commute Hamiltonian-based approach.}
    \\
    \vspace{-0.15cm} 
    \caption{Comparison between \papername and other QAOA designs.}
    \label{fig:intro}
    
\end{figure}

In this work, we propose \papername, a formal and universal framework that smoothly (like a \underline{choco}late) tackles all these limitations. The key novelty of it lies on an expressive representation for encoding the constraints \textemdash ~commute Hamiltonian. In quantum mechanics, Heisenberg's picture states that if the Hamiltonian commutes with an operator, the expectation of the operator remains invariant during the evolution~\cite{heisenbergpicture}. As illustrated in Figure~\ref{fig:intro}~(b), the expectation value of the operator is calculated by projecting the quantum state onto the operator axis. To maintain a constant expectation value, the quantum states are constrained on the plane orthogonal to the operator axis. 
With this understanding, we can deduce that if a Hamiltonian commutes with the constraints operator, the states will be restricted to a subspace where each state adheres to the constraints.
Starting with the commute Hamiltonian, we comprehensively reformulate the QAOA algorithm flow, which is capable of precisely covering arbitrary linear constraints. Furthermore, restricting the evolution under the commute operator substantially shrinks the search space, bringing a great reduction in the number of iterations. 

Then, we present three optimization passes to facilitate end-to-end acceleration and make \papername deployable on current NISQ devices. First, we propose to decouple the commute Hamiltonian into a series of local Hamiltonian, preventing massive tensor computation in the classical part. Unlike the prior Hamiltonian which has to repeatedly execute the Hamiltonian for approximation, we only need to execute one block with much less circuit depth. Though we serialize the entire Hamiltonian into smaller ones, it still has $4^N$ degrees of freedom, leaving a huge unitary decomposition overhead for generating high-quality circuits. To handle this, we design a novel unitary decomposition for the commute Hamiltonian, which achieves linear complexity in both time and circuit depth while keeping strict equivalence. Finally, we employ a variable elimination technique to further reduce the circuit depth from thousands to hundreds. 

The main contributions of this paper include:

\begin{itemize}
    \item We propose a formal, general, and deployable quantum algorithm for constrained binary optimization problems, which leverages the commute Hamiltonian to encode the constraints. This is the first work that accomplishes the goal of a fast and high-confident solver for this problem.

    \item We propose a series of optimization techniques to improve the deployability of the commute Hamiltonian on quantum hardware, including Hamiltonian serialization, equivalent decomposition, and variable elimination. We demonstrate the effectiveness and reliability of these optimizations by introducing two lemmas, which theoretically ensure the linear complexity in both decomposition time and circuit depth.

    \item We conduct rigorous experiments on real-world quantum devices using three practical application scenarios: facility location, graph coloring, and K-partition problem. \papername shows remarkable success rates and achieves end-to-end speedup compared to prior works. 

\end{itemize}

\noindent In algorithmic level, for the problems that prior methods can find the optimal solution, \papername increases the success rate by $235\times$ compared to the state-of-the-art QAOA algorithm that applies cyclic Hamiltonian~\cite{cyclicdriver}. For the problems that prior methods fail to solve, \papername still exhibits 9.5\% - 54\% success rates. On real-world quantum hardware, \papername improves the success rate by $2.65\times$ and achieves a $4.69\times$ speedup compared to \cite{cyclicdriver}. Besides, we reduce the decomposition time over $10^6\times$ and reduce the circuit depth more than $10^4\times$, compared to the existing compilation methods~\cite{trotterbasedsimulation,probabilisticunitarysynthesis,li2019tackling}. \papername is publicly available on \url{https://github.com/JanusQ/Choco-Q.git}.

%% file: tab/relatedwork.tex
\begin{table}[t]
\caption{QAOA designs for constrained binary optimization. }
\vspace{-.3cm}
\label{tab:comparisionwork}
\renewcommand\arraystretch{2.0}
\resizebox{\linewidth}{!}{
\begin{threeparttable}
\begin{tabular}{|l|l|l|l|l|}
\toprule
\textbf{\makecell[l]{Constraint \\ encoding} }
 & \multicolumn{2}{c|}{\textbf{\makecell[c]{Penalty-based}}}  &  \multicolumn{2}{c|}{\textbf{\makecell[c]{Driver-Hamiltonian-based}}}
 % & \multicolumn{2}{c|}{\textbf{\makecell[c]{Objective \\ function}}}  &  \textbf{\makecell[l]{Cyclic \\ Hamilt.}} &  \textbf{\makecell[l]{Commute \\ Hamilt.}}
 \\ \hline
% \cmidrule[0.5pt](rl){0-1} \cmidrule[0.5pt](rl){2-3} \cmidrule[0.5pt](rl){4-5}

\textbf{Methods} &  \makecell[l]{A. Verma \\ et al. \cite{verma2022penalty}} &  \makecell[l]{Red-QAOA \\\cite{redqaoa}} & \makecell[l]{Yoshioka et al.~\cite{cyclicdriver} \\ \textit{cyclic Hamilt.}} & \makecell[l]{\papername \\ \textit{commute Hamilt.}} \\ \hline
 
% \textbf{\makecell[l]{Feature}}  & \makecell[l]{encode objective\\ into Hamiltonian}  &  \makecell[l]{parameter\\ finetune}    & \makecell[l]{cyclic\\ Hamiltonian}  & \makecell[l]{commute\\ Hamiltonian} \\ \hline
% \textbf{\makecell[l]{Encode\\Scheme}}  &  \xmark  &   \xmark   & \xmark      & \makecell[l]{objective\\ Hamiltonian} & \makecell[l]{cyclic \\Hamiltonian} & \makecell[l]{commute\\ Hamiltonian} \\ \hline
\textbf{\makecell[l]{Universality} }  & soft const. & soft const. & \makecell[l]{hard const.\\ only part of linear}  & \makecell[l]{hard const. \\ arbitrary linear}  \\  \hline
% \textbf{\makecell[l]{Constraints\\ satisfaction}}  &  \xmark &   \xmark   & \xmark       & \makecell[l]{soft }& \makecell[l]{ low  } & \makecell[l]{complete}  \\  \hline

\textbf{\makecell[l]{In-constraints \\ rate}}   & 0.03\% & 0.07\%     & 0.67\% & 100\%\\ \hline
 \textbf{\makecell[l]{Success rate}}  &  0.02\% &0.03\% & 0.14\%  & 67.1\%  \\ \hline
 \textbf{\makecell[l]{End-to-end\\ latency}}  &  16.6s & 16.7s & 19.6s & 7.07s  \\
 % \textbf{\makecell[l]{Accuracy \\Level$^*$}}   &  \multicolumn{3}{c}{\makecell{very low}}    &low & low & very high\\
\bottomrule
\end{tabular}
 \begin{tablenotes}
        % \footnotesize
        \item[*] The latency includes computing time and compilation time of QAOA, w/o data communication time. The execution time is estimated based on the IBM Fez model~\cite{ibmq}. 
        % The data of average in-constraints rate, success rate, and end-to-end latency is obtained by experimental test on 100 facility location problems with 21  qubits using the data from~\cite{facilityoverview}.
        % \blackxmark\xspace means the design fails to find the solution.
         % for the 100 problems
      \end{tablenotes}
    \end{threeparttable}
 }
 \vspace{-0.2cm}
\end{table}

%% file: section/2_background.tex
\section{Background}

\subsection{Constrained Binary Optimization}
The constrained binary optimization problem aims to find an optimal assignment of binary variables to minimize or maximize the objective function meanwhile satisfying several constraints. The objective function takes the binary variables as inputs and returns a scalar. The constraints are represented using linear equations, 
\begin{equation}
    \begin{aligned}
        \min_{\vec{x}}~\text{or}~\max_{\vec{x}} &\quad f(\vec{x}),~~~\vec{x}=\{x_1,x_2,\cdots,x_n\}\\
        s.t. &\quad C\vec{x}= \vec{c},~~~\vec{x}\in \{0,1\}^{\otimes n}\\
    \end{aligned}
    \label{eq:cboexample}
\end{equation}
where $C$ is the constraint matrix that contains the coefficients of each linear equation. Figure~\ref{fig:cboexample} (a) gives an example with four binary variables and two constraint equations. And the optimal solution to maximize the objective function is $\{1,0,1,0\}$. The time complexity of classical algorithms grows exponentially with the number of variables and constraints, e.g., the worst-case time complexity is $O(2^n)$ using integer linear programming~\cite{boundbinaryprogramming}.

\subsection{QAOA Preliminaries}\label{subsec:qaoa}
For binary optimization, the QAOA algorithm puts each variable in a qubit and prepares a parameterized quantum circuit to encode the objective function. Generally, it consists of four steps, as shown in Figure~\ref{fig:cboexample}~(b). \textbf{Step 1} initializes the input states through a layer of H gates, which generates a uniform distribution across all possible variable assignments. \textbf{Step 2} constructs the Hamiltonian $H^o$ of the objective function by substituting the variable $x_j$ with $\frac{I_j-Z_j}{2}$ in the objective, where $I_j$ and $Z_j$ are the identity operator and the Pauli Z operator on qubit $j$, respectively. The driver Hamiltonian comprises a set of $R_X$ gates, which is used to parameterize the variable assignments for variational optimization. \textbf{Step 3} simulates Hamiltonian by executing the quantum circuit after compilation. This circuit involves multiple layers of unitary $e^{-i\gamma_l H^o}$ and $R_X(\beta_l)$ gates, where $\gamma_l$ and $\beta_l$ are the parameters for the $l$-th layer. \textbf{Step 4} measures all qubits and generates a sequence of bitstrings representing the variable assignments. Specifically, the QAOA algorithm calculates the results of the objective function under these assignments. Then, it iteratively updates the parameters in Hamiltonian simulation to approximate the optimal solution, e.g., using gradient descent~\cite{gdVQA} or linear approximation method~\cite{cobyla}.

\begin{figure}[t]
    \centering
    % .65
    \includegraphics[width=.92\columnwidth]{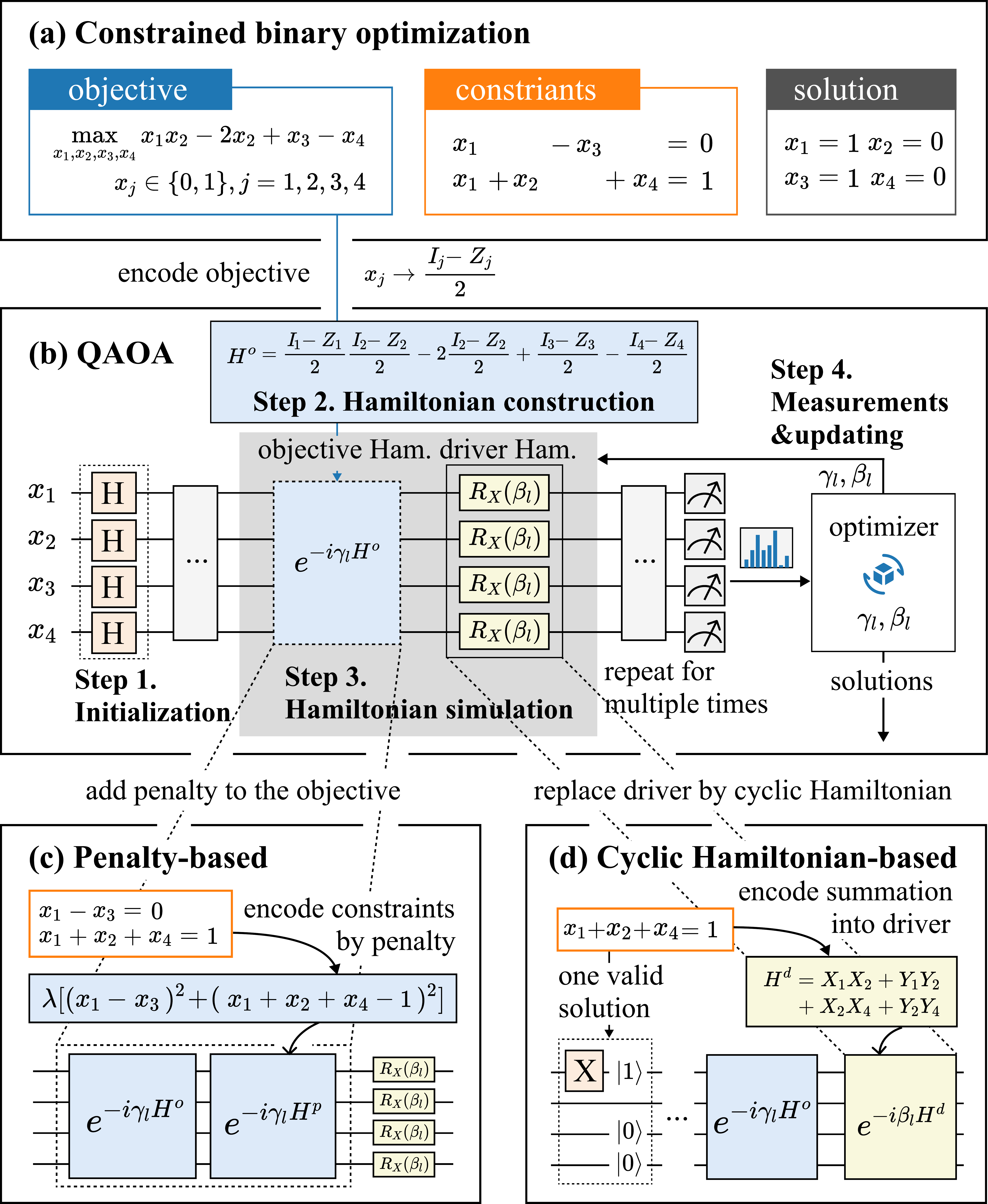}
    \vspace{-0.1cm}
    \caption{Solving a constrained binary optimization using QAOA-based designs. (a) The example of a constrained binary optimization problem. (b) A typical QAOA flow that encodes objective into objective Hamiltonian, without the support of encoding constraints. (c) Extending QAOA by adding penalty terms to solve the constrained problem. (d) Applying cyclic Hamiltonian to encode constraints (only in summation format) into driver Hamiltonian $H^d$.}
    \label{fig:cboexample}
\end{figure}

A natural idea to support constrained binary optimization is to insert the constraints into the objective function as a penalty term~\cite{verma2022penalty} (we call it \emph{soft constraint}). For example, in Figure~\ref{fig:cboexample}~(c), the constraints are formulated by introducing a positive coefficient $\lambda$. 
\begin{equation}
\small
    \begin{aligned}
        \lambda[(x_1-x_3)^2+ (x_1+x_2+x_4-1)^2] \notag
    \end{aligned}
\end{equation}
Theoretically, to ensure the constraints are satisfied, this penalty term should be zero when converging to the optimal solution. During Hamiltonian simulation, an additional penalty Hamiltonian $H^p$ is required, which can be calculated by substituting $x_j=\frac{I_j-Z_j}{2}$ into the penalty term.

% There are two ways to overcome this limitation. One is to integrate the constraint into the objective Hamiltonian. The idea is simple: we can transform the constrained optimization into an unconstrained optimization by adding the square of constraints~(penalty term) into the objective~\cite{verma2022penalty}. For example, in Figure~\ref{fig:cboexample}~(c), the constraint $x_1+x_2=1,x_1-x_3=0$ is added to the objective as a penalty term $\lambda[(x_1+x_2-1)^2+ (x_1-x_3)^2]$, here $\lambda$ is a predefined positive number. Theoretically, the penalty term should become zero when the total objective converges to a minimal value. Thus, the constraints are satisfied. To implement this idea into quantum circuits, we can replace the variable $x_j$ with $\frac{I_j-Z_j}{2}$ in penalty term to get a penalty Hamiltonian $H^p$. The unitary $e^{-i\beta_l H^p}$ is then appended after the objective unitary $e^{-i\beta_l H^o}$ in every layer of the original QAOA circuit. However, as shown in Table~\ref{tab:comparisionwork}, QAOA-penalty suffers from a long circuit and a high error rate. This is because the penalty term is a soft constraint and the output solution can not always satisfy the constraints. 

An alternative approach is to design a new driver Hamiltonian for the constraints (we call it \emph{hard constraint}). For example, the cyclic Hamiltonian~\cite{cyclicdriver} is able to deal with the summation format of the constraints, 
which is inspired by the one-dimensional Ising model~\cite{isingmodel}.
\begin{equation}
\small
    \begin{aligned}
        H^d=\sum_{i=1}^{n-1} X_{i}X_{(i+1)} + Y_{i}Y_{(i+1)}
    \end{aligned}
    \label{eq:cyclicham}
\end{equation}
Here $X_i$ and $Y_i$ are the Pauli-X and Pauli-Y operators on qubit $q_i$, respectively, which are determined by the variables in constraint. For example, the constraint $x_1+x_2+x_4=1$ involves variables $x_1,x_2,x_4$, and its cyclic Hamiltonian is
\begin{equation}
    \small H^d= X_1X_2+Y_1Y_2 +X_2X_4+Y_2Y_4 
    \notag    
\end{equation}
In addition, we need to set the initial state to the solution of the constraint equation. For example, one solution of the constraint equation $x_1+x_2+x_4=1$ is $x_1=1,x_2=0,x_4=0$, thereby setting the three qubits to $\ket{100}$ as shown in Figure~\ref{fig:cboexample}~(d).

%% file: section/3_formulation.tex
\section{\papername Formulation}

·Compared to penalty-based methods, encoding to driver Hamiltonian allows hard constraints and higher accuracy. However, the cyclic Hamiltonian is fundamentally limited by the summation format of the constraints that all coefficients must have the same sign, e.g., $x_1+x_3=1$, $-x_1-x_2=-1$, and different constraint equations cannot share the same variables. Such limitation arises from the fact that the cyclic Hamiltonian, derived from the Ising model, has to keep the number of excited states $\ket{1}$ constant to encode the constraint. 

Inherently, QAOA is a discrete simulation of the quantum evolution under objective Hamiltonian and driver Hamiltonian. Furthermore, for a linear constraint $\sum_{i=1}^n{c_ix_i}=c$, its operator is defined as,
\begin{equation}\label{eq:cons_op}
    \small
    \hat{C}=\sum_{i=1}^n{c_i\sigma^z_i}
\end{equation}
where the driver Hamiltonian is responsible for maintaining the expected value of this operator constant during evolution. On the other hand, the Heisenberg picture states that if the driver Hamiltonian commutes with the operator, its expected value of the operator will remain unchanged~\cite{heisenbergpicture}. With this understanding, we propose to employ the commute operator as driver Hamiltonian, namely commute Hamiltonian, which turns out to be a more general QAOA algorithm that enables arbitrary linear constraints for binary optimization, facilitating accuracy, scalability, and deployability.

\subsection{Applying Commute Hamiltonian to QAOA}
\label{sec:commutehamiltonian}
The commute Hamiltonian supports hard constraints that can always satisfy the constraints. Commute Hamiltonian originates from the Heisenberg picture~\cite{heisenbergpicture}, which is an equal expression of the Schrödinger equation~\cite{eqschandheisenberg}. The Heisenberg picture gives that for any operator $\hat{A}$, the variation rate $\frac{d \hat{A}}{d t}$ under the evolution of Hamiltonian $H$ is proportional to the commutation operator between $\hat{A}$ and $H$. 
\begin{equation}
\small
    \begin{aligned}
    & \frac{d \hat{A}}{d t} = \frac{i}{\hbar}[\hat{A},H] \\
    & [\hat{A},H]= \hat{A}H-H\hat{A}
    \end{aligned}
    \label{eq:commute}
\end{equation}
Here, $[\hat{A},H]$ is the commutation operator. $[\hat{A}, H]=0$ leads to $\frac{d \hat{A}}{d t} = 0$, which implies that the expected value of $\hat{A}$ will not change during the evolution of $H$. In other words, we can find a driver Hamiltonian $H^d$ that commutes with the constraint operator $\hat{C}$ for constrained binary optimization $[\hat{C}, H^d] =0$. 
In QAOA, the Hamiltonian $H$ comprises the objective Hamiltonian $H^o$ and the driver Hamiltonian $H^d$, denoted as $H = H^o + H^d$. Note that arbitrary objective Hamiltonian $H^o$ commutes with $\hat{C}$. This is because $H^o$ only involves $I$ and $\sigma^z$ operators, which always commute with the $\sigma^z$ operator in $\hat{C}$.
Next, our goal is to find the driver Hamiltonian $H^d$ that also commutes with $\hat{C}$. Hannes et al.~\cite{constructHd} gives a universal equation to find the commute Hamiltonian for specific constraints.
\begin{equation}
\small
\begin{aligned}
& H^d = \sum_{\vec{u}\in \Delta}H_{c}(\vec{u})= \sum_{\vec{u}\in \Delta} 
(\sigma_1^{u_1}\cdots \sigma_n^{u_n}+ 
% \sum_{\vec{u}\in \Delta} 
\sigma_1^{-u_1}\cdots \sigma_n^{-u_n} )\\
& C\vec{u}=\vec{0},\;\vec{u} = \{u_1,u_2,u_3,...,u_n\}, ~u_i\in \{1,0,-1\} \\
& \sigma_i^{+1}= \begin{bmatrix}
0  & 0\\
1 & 0
\end{bmatrix},~~\sigma_i^{0}= \begin{bmatrix}
1  & 0\\
0 &1
\end{bmatrix},~~\sigma_i^{-1}= \begin{bmatrix}
0  & 1\\
0 &0
\end{bmatrix}
\end{aligned}
\label{eq:hd}
\end{equation}
where $\Delta$ is the set of all valid solutions of $\vec{u}$ that satisfies $C\vec{u}=0$, and $H_c(\vec{u})$ is the commute Hamiltonian for each valid $\vec{u}$.

\begin{figure}[t]
    \centering
    \includegraphics[width=.95\linewidth]{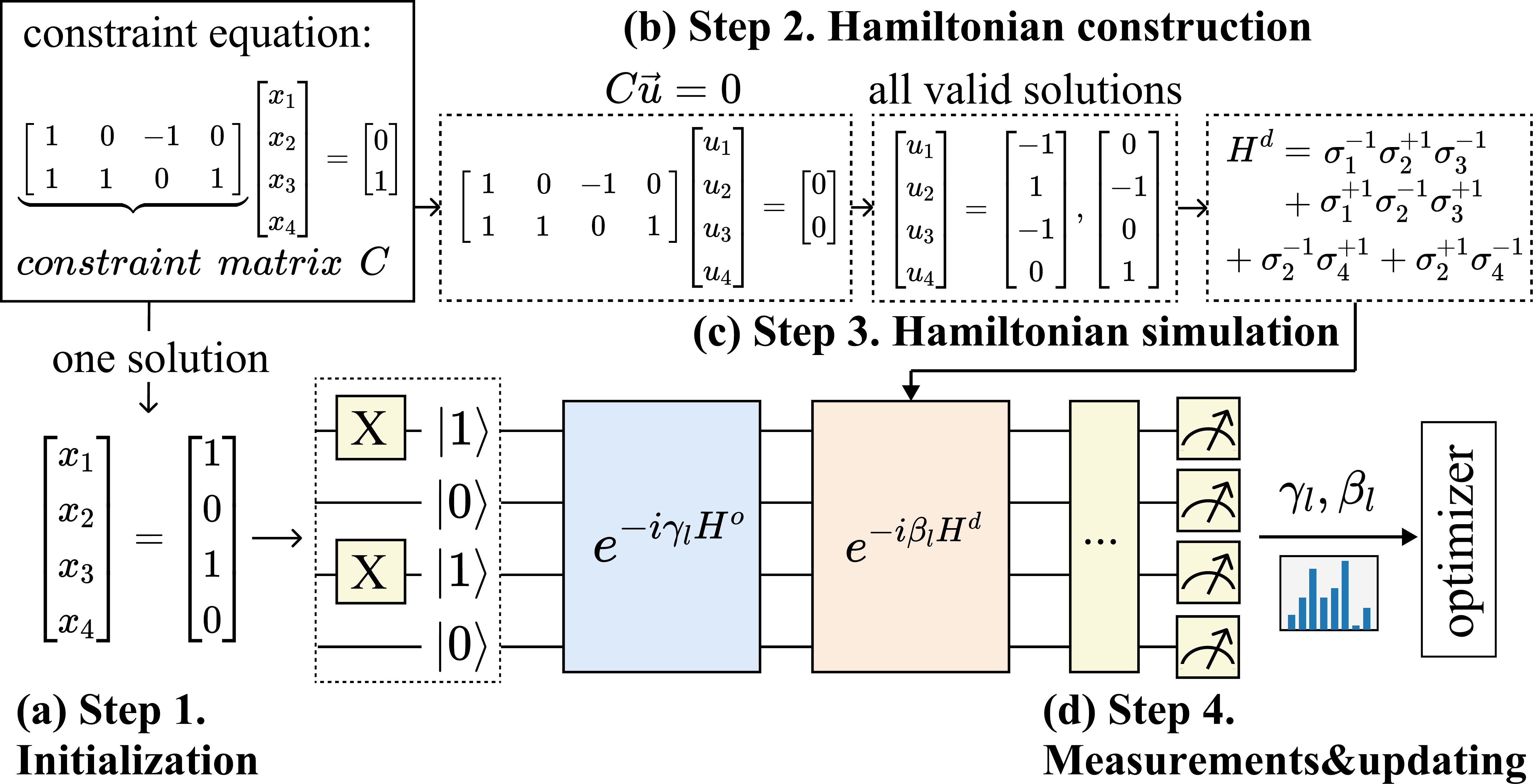}
    \vspace{-0.1cm}
    \caption{The workflow of encoding the constraints via commute Hamiltonian.}
    \label{fig:commuteflow}
\end{figure}

Figure~\ref{fig:commuteflow} depicts the basic flow of applying commute Hamiltonian to QAOA. Similarly, the initial quantum states are set to one solution of the constraint equation. The main difference lies in the design of the commute Hamiltonian and its simulation. Since the constraints are encoded by setting the commutation operator to zero $[\hat{C}, H^d] =0$, all solutions of $C\vec{u}=0$ form the commute Hamiltonian. For example, $\vec{u}^{1}=[-1,1,-1,0]$ and $\vec{u}^{2}=[0,-1,0,1]$ are all valid solutions for the case in Figure~\ref{fig:commuteflow}~(b), resulting in the following Hamiltonian according to Equation \eqref{eq:hd}.
\begin{equation}
\small
\begin{aligned}\label{eq:hdexample}
H^d & = H_{c}(\vec{u}^{1}) + H_{c}(\vec{u}^{2}) \\
 H_{c}(\vec{u}^{1}) &=\sigma_1^{-1} \sigma_2^{+1} \sigma_3^{-1}
+\sigma_1^{+1} \sigma_2^{-1} \sigma_3^{+1} \\
H_{c} (\vec{u}^{2}) &= \sigma_2^{-1} \sigma_4^{+1}+\sigma_2^{+1} \sigma_4^{-1}
\end{aligned}
\end{equation}
Consequently, the Hamiltonian simulation process is formulated by repeatedly executing the objective Hamiltonian $H^o$ and the commute Hamiltonian $H^d$.
\begin{equation}\label{eq:circuit}
\small
\begin{aligned}
& \ket{\phi_{\theta}} = \underbrace{e^{-i\beta_L H^d}e^{-i\gamma_L H^o} \cdots e^{-i\beta_1 H^d}e^{-i\gamma_1 H^o}}_{\text{repeat for}\;L\;\text{times}} \ket{\vec{x^*}} \\
& \theta = \{\gamma_l,\beta_l\}^L_{l=1} = \{\beta_1,\gamma_1,\beta_2,\gamma_2,...,\beta_L,\gamma_L\} 
\end{aligned}
\end{equation}
Here, $\ket{\vec{x^*}}$ is the initial state in step 1 (e.g, it is $\ket{1010}$ Figure~\ref{fig:commuteflow}~(a)), $\ket{\phi_{\theta}}$ is the final state, and $\theta$ denotes $2L$ adjustable parameters for approximate optimization. From the perspective of commute Hamiltonian, the cyclic Hamiltonian can be regarded as a special case ($c_i \equiv 1$ in Equation \eqref{eq:cons_op}).

\subsection{Advantages and Challenges}

Commute Hamiltonian is a much more general encoding scheme that deals with arbitrary linear constraints. By restricting the evolution under the commute operator, the search space significantly shrinks, resulting in a high-quality solution. More importantly, compared to the penalty-based QAOA, which often requires all-to-all entanglement between qubits, commute Hamiltonian exhibits sparse connectivity, making it hardware-friendly.

\emph{Though the commute Hamiltonian shows higher generality and preciseness, it inevitably introduces computational overhead to obtain the accurate Hamiltonian.} According to Equation \eqref{eq:hd},
% this Hamiltonian necessitates the complete solution space of $C\vec{u}=0$, which takes $O(n^3)$ complexity using classical solver like Gauss elimination~\cite{gaussianelimination}. Besides this
the calculation of the commute Hamiltonian involves massive tensor-product and accumulation, further aggravating the overwhelming computational pressure. For $n$ qubits Hamiltonian, the calculation complexity is $O(2^n)$ and the time complexity is $O(n2^n)$.

\emph{Like other QAOA approaches, the Hamiltonian unitary usually leads to huge circuit complexity after decomposition, making it challenging to deploy on NISQ devices.} Typically, existing decomposition methods show exponential complexity to approximate the unitary~\cite{chen2022optimizing,tan2023quct}, and often exhibit high approximation error~\cite{xu2023synthesizing}. For example, trotter decomposition\cite{trotterbasedsimulation} approximates the unitary $e^{-i\beta H^d}$ by dividing it into a series of small unitaries $e^{-i\beta H^d/N}$.
\begin{equation}\label{eq:trotter}\small
e^{-i\beta H^d} =  \underbrace{e^{-i \beta H^d/N}\cdots e^{-i\beta H^d/N}}_{\text{repeat for N times} }
\end{equation}
However, the accuracy of trotter decomposition depends on the number of $e^{-iH^d\beta/N}$ repetitions (the error is $O(1/N^2)$ after repeating N times). Moreover, these small unitaries still have to be decomposed into basic gates, severely increasing the circuit complexity and making it undeployable.

%% file: section/4_opt2.tex
\section{\papername Optimization}

To address the new challenges introduced by commute Hamiltonian, we propose multiple optimization passes. 
% as shown in Figure~\ref{fig:workflow}. 
To deal with the overhead of tremendous classical computation when calculating the driver Hamiltonian, we serialize the commute Hamiltonian into a set of local Hamiltonian. Unlike the conventional simulation method that exhaustively partitions the Hamiltonian and repeatedly executes them, our serialization technique only needs to calculate the local Hamiltonian that takes less computational cost (Section~\ref{subsec:seqcommte}). Each local commute Hamiltonian still requires to be decomposed into basic gates for deployment. Different from prior approximation-based unitary decompositions~\cite{probabilisticunitarysynthesis,trotterbasedsimulation}, we present an equivalent transformation that decomposes the commute Hamiltonian into several control gates and phase gates. Using this, we demonstrate the linear complexity in both decomposition time and circuit depth (Section~\ref{subsec:phasegate}). By putting these two techniques together, we effectively prevent around GFlops tensor computation and reduce the circuit depth from {$10^{10}$} to $\sim$1000. To further improve deployability on the current NISQ devices, we propose a variable elimination technique that can squeeze the constraint matrix, shrinking the circuit complexity (Section~\ref{subsec: divide}). Finally, we successfully achieve an executable circuit depth, around {100-depth}, meanwhile exhibiting high accuracy for solving the constraint binary optimization problem.

\subsection{Serialization of Commute Hamiltonian}
\label{subsec:seqcommte}

The commute Hamiltonian in Equation~\eqref{eq:hd} exhibits highly-entangled complexity that requires connecting a large number of qubits, making it hard to deploy on sparsely-connected quantum chips. A natural idea is to reduce it into a series of smaller Hamiltonian, namely \textit{local Hamiltonian} that acts only on a few qubits. However, note that when $H_1$ and $H_2$ are complex matrices, $e^{H_1+H_2} \neq e^{H_1}e^{H_2}$, which means that each item $e^{-i\beta H^d}$ in the driver Hamiltonian (Equation~\eqref{eq:circuit}) cannot be directly separated. 
\begin{equation}
    \small
    e^{-i\beta H^d} = e^{-i\beta \sum_{\vec{u}\in \Delta} H_c(\vec{u})} \neq \Pi_{\vec{u}\in \Delta} e^{-i\beta  H_c(\vec{u})}
    \notag
\end{equation}
Taking $\vec{u}^1= [-1,0], \vec{u}^2= [-1,1], \beta = 0.8$ as an example, we can easily verified that $e^{-i\beta (H_c(\vec{u}^1)+ H_c(\vec{u}^2))} \neq e^{-i\beta H_c(\vec{u}^1)}e^{-i\beta H_c(\vec{u}^2)}$. Essentially, the function of the commute Hamiltonian lies on the encoding of the constraints. In other words, though the separation of Hamiltonian fails to keep the equivalence, we prove that it still meets the constraint operator of Equation \eqref{eq:cons_op}, which provides the opportunity to shrink the circuit complexity.

\begin{lemma}\label{theorem:driverseq}
    Replacing the commute Hamiltonian unitary $e^{-i\beta H^d}$ with the serialization of unitaries $\Pi_{\vec{u}\in \Delta} e^{-i\beta  H_c(\vec{u})}$ still satisfies the constraint operator during evolution. 
    \begin{equation}
    \small
    \begin{aligned}
     & \ket{x_c}= e^{-i\beta H^d} \ket{x},~\ket{x_s}= \Pi_{\vec{u}\in \Delta} e^{-i\beta H_c(\vec{u})} \ket{x}\notag \\
     & \bra{x_s}\hat{C} \ket{x_s} = \bra{x_c}\hat{C} \ket{x_c}
    \end{aligned}
    \end{equation}
    % \vspace{-0.3cm}
    % Replacing the commute Hamiltonian unitary $e^{-i\beta H^d}$in the original circuit (Equation~\eqref{eq:circuit}) with the serialization of unitaries $e^{-i\beta H_c(\vec{u})}$ can also maintain the expected value of the constraints operator  $\hat{C}=\sum_{i=1}^n{c_i\sigma^z_i}$ constant during evolution, which means the measured bitstrings $\{x_i\}_{i=1}^{n}$ always satisfying the constraints $\sum_{i=1}^n{c_i x_i}=c$.
    % \begin{equation}
    % \small
    % \begin{aligned}
    %  e^{-i\beta H^d} &\longrightarrow  \Pi_{\vec{u}\in \Delta} e^{-i\beta  H_c(\vec{u})} \\
    % % H_{c}(\vec{u}) &= \sigma_1^{u_1}\cdots \sigma_n^{u_n}+\sigma_1^{-u_1}\cdots \sigma_n^{-u_n}
    % \end{aligned}
    % \end{equation}
\end{lemma}

\begin{figure}[t]
    \centering
\includegraphics[width=0.98\linewidth]{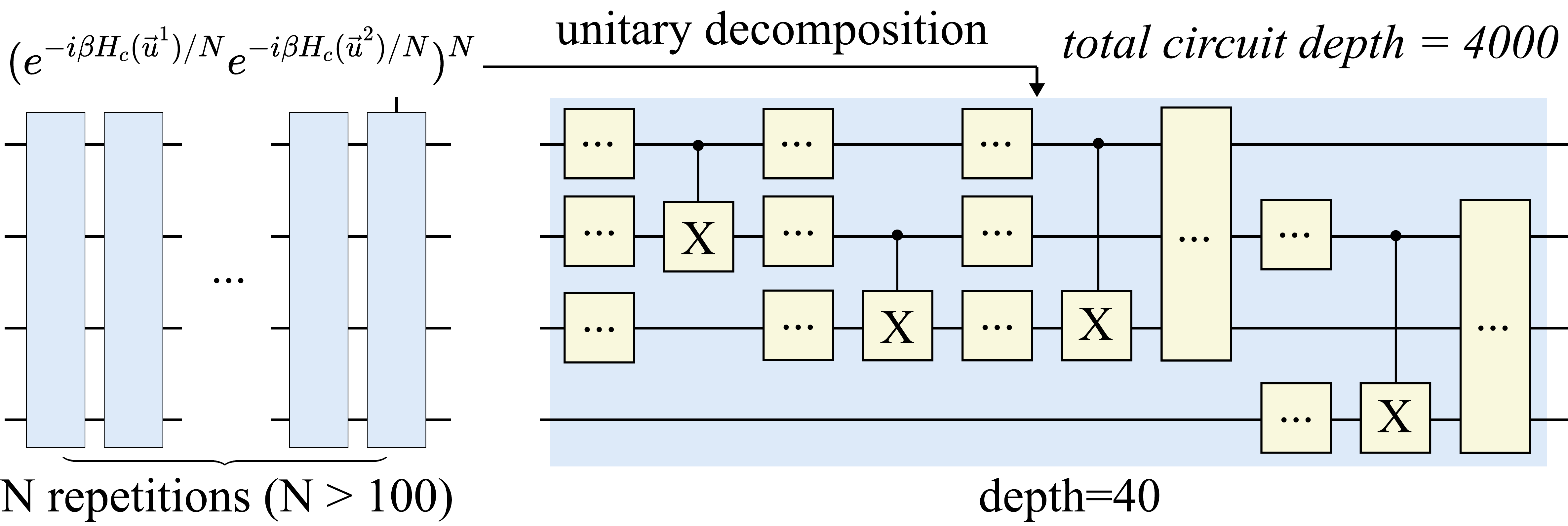}
     \\
     \vspace{-0.1cm}
    \makebox[0.90\linewidth]{\footnotesize (a) Trotter decomposition~\cite{trotterbasedsimulation}.}
    \vspace{0.25cm} 
    \\
    \includegraphics[width=.98\linewidth]{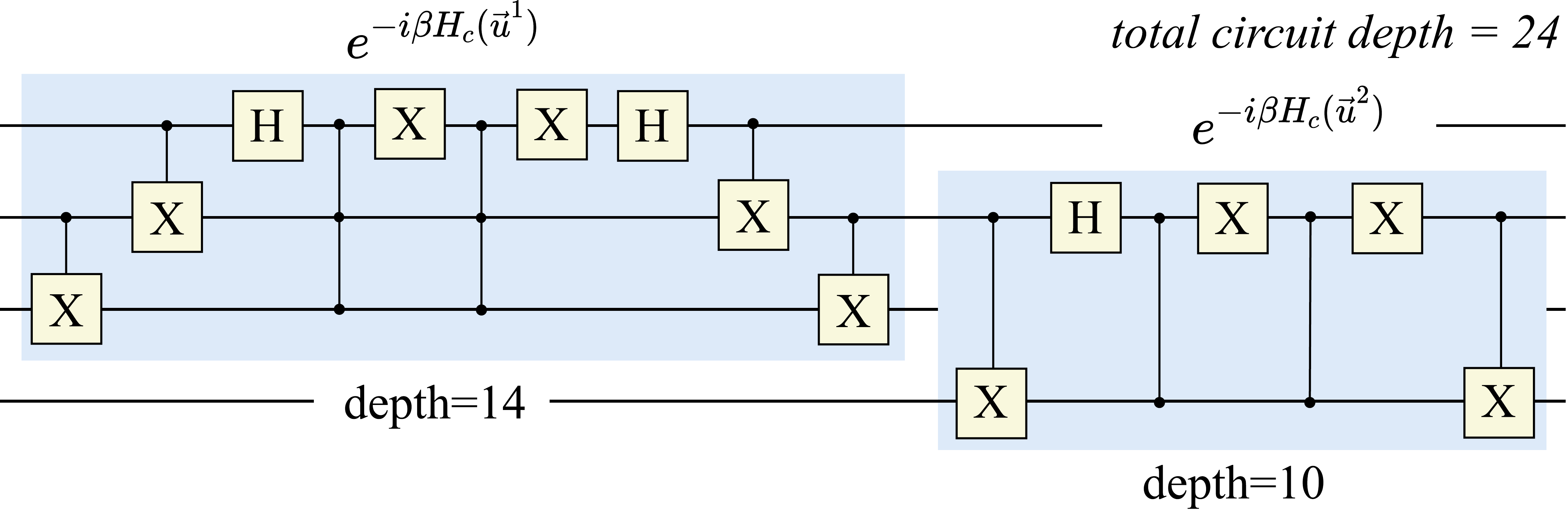}
    \\
    \vspace{-0.1cm} 
    \makebox[0.90\linewidth]{\footnotesize (b) Serialization of commute Hamiltonian.}
    \\
    \vspace{-0.15cm} 
    \caption{Different Hamiltonian simulation methods for Equation \eqref{eq:hdexample}.}
    \label{fig:seqdriver}
\end{figure}

\begin{proof}
    The expected value of the constraints operator $\hat{C}$ is calculated by the inner product $\bra{x}\hat{C}\ket{x}$, where $\ket{x}$ is the initial state. The quantum state $\ket{x_c}$ after applying the commute Hamiltonian unitary satisfies the constraint during evolution:
     \begin{equation}\small
       \bra{x_c}\hat{C} \ket{x_c} = \bra{x}\hat{C} \ket{x}\label{proof1:commuteh}
    \end{equation}
    Then, we use $\hat{B}$ to denote $\Pi_{\vec{u}\in \Delta} e^{-i\beta H_c(\vec{u})}$. The expected value of $\hat{C}$ after applying this serialized unitaries $\hat{B}$ can be calculated as follows,
    \begin{equation}\small
    \begin{aligned}
    \bra{x_s}\hat{C} \ket{x_s} &= \bra{\hat{B}x} \; \hat{C}\ket{\hat{B} x}  = \bra{x} \hat{B}^{\dagger} \;\hat{C}\hat{B} \ket{x}  \\
    \end{aligned}\label{proof1:eq1}
    \end{equation}
    where $\hat{B}^{\dagger}$ denotes the conjugate transpose of $\hat{B}$, satisfying $\hat{B}^{\dagger} \hat{B}$ is the identity operator. On the other hand, we can easily verify that each Hamiltonian $H_c(\vec{u})$ commutes with $\hat{C}$.
    \begin{equation}\small
         [H_c(\vec{u}),\hat{C} ] =0 \notag
         \label{proof1:commutehc}
    \end{equation}
    Furthermore, since the Hamiltonian unitary can be expanded by its Taylor series,
   \begin{equation}\footnotesize
         e^{-i\beta H_c(\vec{u})} =  {\textstyle \sum_{k=1}^{\infty} } (-i \beta H_c(\vec{u})/k!)^k\notag
    \end{equation}
    and we have proven that each term $H_c(\vec{u})$ all commute with $\hat{C}$,
    the Hamiltonian unitaries then commute with $\hat{C}$.
    Thus, their product also commutes with $\hat{C}$, which implies $\hat{C} \hat{B}=  \hat{B} \hat{C}$. Substituting it into Equation \eqref{proof1:eq1}, we have
 %    \begin{equation}\small
 %         % [\hat{C},\hat{B}]=0 \Longleftrightarrow 
 %          \label{eq:cbcommute}
 %    \end{equation}
 %    By this commutation, we can exchange the order of $\hat{B}$ and $\hat{C}$ in Equation~\eqref{proof1:eq1} to get
 %    \begin{equation}\small
 %    \begin{aligned}
 % \bra{x} \hat{B}^{\dagger} \;\hat{C}\hat{B} \ket{x} =\bra{x} \hat{B}^{\dagger} \hat{B} \hat{C} \ket{x}\notag
 %    \end{aligned}
 %    \notag 
 %    \end{equation}
 %    By Equation~\eqref{proof1:Bdagger}, $\hat{B}^{\dagger}\hat{B}$ equals the identity operator $\hat{I}$ and The identity operator $I$ does not affect arbitrary operator $I\hat{C}=\hat{C}$.
 %    \begin{equation}\small
 %         \bra{x} \hat{B}^{\dagger} \hat{B} \hat{C} \ket{x} = \bra{x}I\;\hat{C}\ket{x} = \bra{x} \hat{C} \ket{x}
 %    \end{equation}
 %    Combined with Equation~\eqref{proof1:commuteh}, the serialization of unitaries thus also satisfies the constraint operator.
    \begin{equation}\small
    \begin{aligned}
    \bra{x_s}\hat{C} \ket{x_s} =\bra{x} \hat{B}^{\dagger} \hat{C}\hat{B} \ket{x}  = \bra{x} \hat{B}^{\dagger} \hat{B} \hat{C}\ket{x}  = \bra{x_c}\hat{C} \ket{x_c}
    \end{aligned}\notag
     \vspace{-0.65cm}
    \end{equation}
\end{proof}
\vspace{-0.1cm}

Lemma 1 gives the theoretical foundation for the serialization of commute Hamiltonian, which is specific to the constraint binary optimization problem. Compared to conventional Hamiltonian simulation methods~\cite{hamiltoniansimulationbyqsp}, our approach significantly reduces the circuit depth from thousands to dozens, making it possible to implement the circuit. For example, Figure~\ref{fig:seqdriver}~(a) shows trotter decomposition~\cite{trotterbasedsimulation} that is widely used for implementing the driver Hamiltonian (Equation \eqref{eq:trotter}). This method cuts the Hamiltonian into N pieces and repeatedly executes them, leading to tremendous gates and huge circuit depth. In contrast, Figure~\ref{fig:seqdriver}~(b) depicts the circuit that embeds the constraints using the serialization of commute Hamiltonian. We effectively eliminate the repetition of massive small Hamiltonian, reducing the circuit depth from 4000 to 24 (after decomposition). Moreover, our method exhibits less connection between qubits, which is more hardware-friendly when compiling the circuit for a certain topology. For instance, $H_c(\vec{u}^1)$ acts on three qubits $q_1,q_2,q_3$, and $H_c(\vec{u}^2)$ only involves two qubits $q_2,q_4$.

\subsection{Decomposition of Commute Hamiltonian}
\label{subsec:phasegate}

Though we serialize the large commute Hamiltonian into a set of smaller Hamiltonian, it still requires decomposing these unitaries into basic gates. In particular, the decomposition quality highly determines the final accuracy and the deployability, necessitating a precise and fast methodology. To achieve this, we identify that applying the local commute Hamiltonian to its eigenstate fundamentally behaves like a phase gate, 
\begin{equation}\small
    \begin{aligned}
       & e^{-i\beta H_c(\vec{u})} \ket{x^{\pm}} =   e^{\mp i\beta}\ket{x^{\pm}} \\ 
       & e^{-i\beta H_c(\vec{u})} \ket{other} = \ket{other} 
    \end{aligned}
\end{equation}
where $|x^{\pm}\rangle$ are the eigenstates. 
\begin{equation}\small
    \begin{aligned}
       & |x^{\pm}\rangle = \frac{|v_1\cdots v_n \rangle \pm  |\bar{v}_1\cdots \bar{v}_n\rangle}{\sqrt{2}} \\ 
       & v_i = \frac{1+u_i}{2},\bar{v}_i = \frac{1-u_i}{2}
    \end{aligned}
    \label{eq:x_state}
\end{equation}
Inspired by this property, we first introduce an equivalent transformation that represents the Hamiltonian into multi-phase control gates. Then, we demonstrate the linear complexity in both time and circuit depth.

% The decomposition of each item in the serialization of commute Hamiltonian 

% Even if we serialize the large Hamiltonian into a smaller Hamiltonian, direct decomposition of the Hamiltonian unitary still turns out to be challenging and suffers from low accuracy. We proposed an exact decomposition approach that leverages the eigenstate properties of the local commute Hamiltonian unitary. We discover that the effect of local commute Hamiltonian applying on its eigenstate is like a phase gate. The phase gate does not change the eigenstate but adds a global phase to it. This inspires us to develop a decomposition method using phase gates. Since the result of eigenstate decomposition is unique for arbitrary unitary, this approach gives an exact decomposition for the local commute Hamiltonian unitary. The following theorem gives an exact and efficient decomposition method for the local commute Hamiltonian unitary.

\begin{lemma}\label{theorem:phasedecomp}
    Each commute Hamiltonian $e^{-i\beta H_c(\vec{u})}$ can be equivalently decomposed as follows.
    \begin{equation}
    \label{eq:decomposition}
    \small
        e^{-i\beta H_c(\vec{u})}= G^\dagger P(\beta)X_1P(-\beta)X_1G 
    \end{equation}
    Here, $X_1$ gate refers to an X gate applied on the first qubit $q_1$, $G$ is a set of gates to convert the eigenstate $\ket{x^{+}}$ and $\ket{x^{-}}$  into basis state $\ket{01\cdots1}$ and $\ket{11\cdots1}$, respectively.
    \begin{equation}\label{eq:ggate}
    \small
    \begin{aligned}
        & G\ket{x^{+}} = \ket{01\cdots1},G\ket{x^{-}} =\ket{11\cdots1}
    \end{aligned}
    \end{equation}
    $P(\beta)$ is a multi-control phase gate that only adds a phase $e^{i\beta}$ to the state $\ket{1\cdots1}$ and keeps the other eigenstates as same. 
    \begin{equation}
    \label{eq:phasegate}
    \small
    \begin{aligned}
        & P(\beta) \ket{1\cdots1} = e^{i\beta} \ket{1\cdots1} \\ 
        & P(\beta) \ket{other} = \ket{other} 
    \end{aligned}
    \end{equation}
\end{lemma}

% \begin{lemma}\label{theorem:phasedecomp}
%     The unitary of local commute Hamiltonian $e^{-i\beta H_c(\vec{u})}$ can be decomposed by the following circuit with $O(n)$ time complexity and $O(n)$ circuit depth ($n$ is the number of qubits that the unitary acts on).
% \begin{equation}\label{eq:decomposition}
%     \small
%         e^{-i\beta H_c(\vec{u})}= G^\dagger P(\beta)X_1P(-\beta)X_1G
%     \end{equation}
%     Here, $G$ is the gates to convert the eigenstate $\ket{x^{+}}$ and $\ket{x^{-}}$  into basis state $\ket{01\cdots1}$ and $\ket{11\cdots1}$, respectively.
%     \begin{equation}
%     \small
%     \begin{aligned}
%         &G\ket{x^{+}} = \ket{01\cdots1},G\ket{x^{-}} =\ket{11\cdots1}\\
%         |x^{\pm}\rangle = &\frac{|v_1\cdots v_n \rangle \pm  |\bar{v}_1\cdots \bar{v}_n\rangle}{\sqrt{2}}, v_i = \frac{1+u_i}{2},\bar{v}_i = \frac{1-u_i}{2}\notag
%        \label{eq:ggate}
%     \end{aligned}
%     \end{equation}
%     Besides, $G^\dagger$ is the invert operation of $G$, which can be implemented by applying the gates in $G$ in a reversed order. $P(\beta)$ is a multi-control phase gate and it only adds a phase $e^{i\beta}$ to the state $\ket{1\cdots1}$. It means that
%     \begin{equation}\label{eq:phasegate}
%     \small
%         P(\beta) \ket{1\cdots1} = e^{i\beta} \ket{1\cdots1},\;
%         P(\beta) \ket{other} = \ket{other} 
%     \end{equation}
%      Similarly, $P(-\beta)$ is a multi-control phase gate with phase $e^{-i\beta}$. The $X_1$ gate is the X gate applied on the first qubit $q_1$.
% \end{lemma}

\begin{proof}
    We prove this lemma by demonstrating that the decomposition in Equation~\eqref{eq:decomposition} maintains the same eigenstates $|x^{\pm}\rangle$. According to Equation~\eqref{eq:ggate}, taking $\ket{x^{+}}$ as an example,
    \begin{equation}
    \small
    \begin{aligned}
        & G^\dagger P(\beta)X_1P(-\beta)X_1G \ket{x^{+}} \\
    =   & G^\dagger P(\beta)X_1P(-\beta)X_1\ket{01\cdots1}
    \end{aligned}\notag
    \end{equation}
    applying $X_1$ gate on qubit $q_1$ will change $\ket{01\cdots1}$ to $\ket{11\cdots1}$. According to Equation~\eqref{eq:phasegate}, the phase gate $P(-\beta)$ will put a phase $e^{-i\beta}$ on this state.
    \begin{equation}
    \small 
       P(-\beta)X_1 \ket{01\cdots1} = e^{-i\beta} \ket{11\cdots1} \notag
    \end{equation}
    Similarly, according to the definition of $G^{\dagger}$ and $P(\beta)$, we can verify that,
    \begin{equation}
    \small
        G^\dagger P(\beta)X_1P(-\beta)X_1G \ket{x^{+}}= e^{-i\beta} \ket{x^{+}} \notag
    \end{equation}    
    Furthermore, since the other states cannot activate the multi-control phase gate this decomposition will keep the other eigenstates unchanged.
    \begin{equation}
    \small
        G^\dagger P(\beta)X_1P(-\beta)X_1G \ket{other}=\ket{other} \notag
     \vspace{-0.65cm}
    \end{equation}
\end{proof}
\vspace{-0.2cm}

\input{algorithms/basisconvert}

\begin{figure*}[t]
\centering
\includegraphics[width=.98\textwidth]{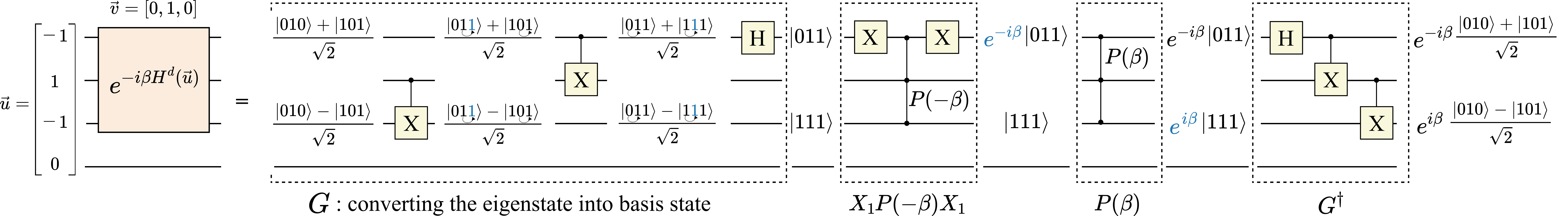}
    \vspace{-0.2cm}
    \caption{The end-to-end decomposition flow of the commute Hamiltonian $e^{-i\beta H_c(\vec{u}^1)}$, following the same example in Equation \ref{eq:hdexample}.}
    \label{fig:phasegatedecomp}
    \vspace{-0.3cm}
\end{figure*}

After decomposing the commute Hamiltonian into several convert gates $G$ and phase gates $P(\beta)$, we then illustrate that compiling these gates to basic gates (e.g., Hadamard, $R_Z$, CX) only takes linear time complexity and linear circuit depth. The implementation of the convert gate $G$ is shown in Algorithm~\ref{alg:Ggate}. First, the eigenstate $\ket{x^{\pm}}$ is calculated using $v_i$ in Equation~\ref{eq:x_state} (line 1-3). The core idea is to transform the last $n-1$ qubits to $\ket{1}$ state (line 4), which is implemented by CX and X gates (lines 5-10). Clearly, the CX gate helps to flip the target qubit if the control qubit is $\ket{1}$, while the X gate flips the state when two qubits are in the same state. After this, the quantum state is manipulated to a reduced state $\ket{s^{\pm}}$.
\begin{equation}
        \small
        \ket{s^{\pm}} = (\ket{0}\pm \ket{1})\ket{1\cdots1}/\sqrt{2}
\end{equation}
Finally, we use an H gate on the first qubits to get $\ket{01\cdots1}$ for $|x^{+}\rangle$ and $\ket{11\cdots1}$ for $\ket{x^{-}}$, respectively. Overall, according to Algorithm~\ref{alg:Ggate}, the time complexity and circuit depth are with $O(n)$ complexity that is linear to the number of qubits in commute Hamiltonian.

The decomposition of the multi-quit phase gate can be visualized as the following equation. 
\begin{figure}[!h]
\centering
\vspace{-0.4cm}
\includegraphics[width=.95\columnwidth]{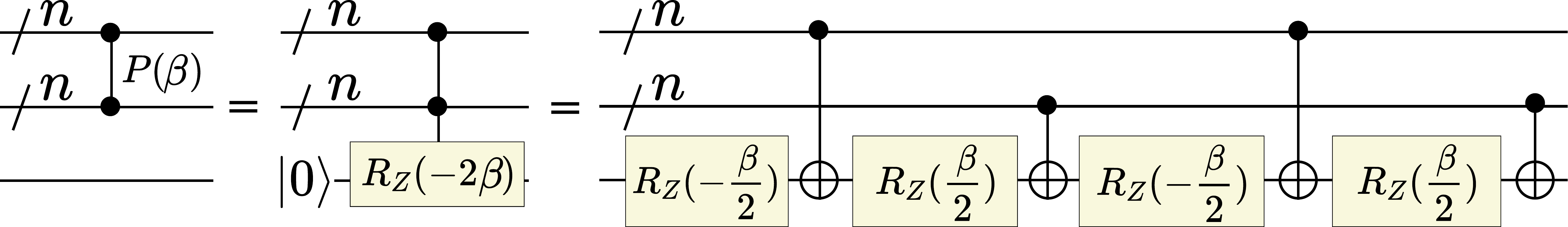}
    % \caption{The decomposition of multi-qubit phase gate.}
    % \label{fig:phasecomp}
    \vspace{-0.35cm}
\end{figure}
\noindent We first reformulate the phase gate $P(\beta)$ as an $R_Z$ gate controlled by $2n$ qubits with an ancillary qubit in $\ket{0}$. Next, the control-$R_Z$ gate, with $2n+1$ qubits, can be further implemented by four control-X gates and four $R_Z$ gates, with each control-X gate involving $n+1$ qubits. Since the control-X gates can be implemented with $16n$-qubit gate with one ancillary qubit~\cite{gidneytoffoli}, the overall complexity of decomposing the $n$-qubit phase gate is $O(n)$. On the other hand, there are two ancillary qubits in total, which can be reused in the entire circuit of commute Hamiltonian simulation. Thus, the overall circuit complexity is also $O(n)$ with only two ancillary qubits.

\begin{figure}[t]
\vspace{-0.3cm}
\centering
\includegraphics[width=.98\columnwidth]{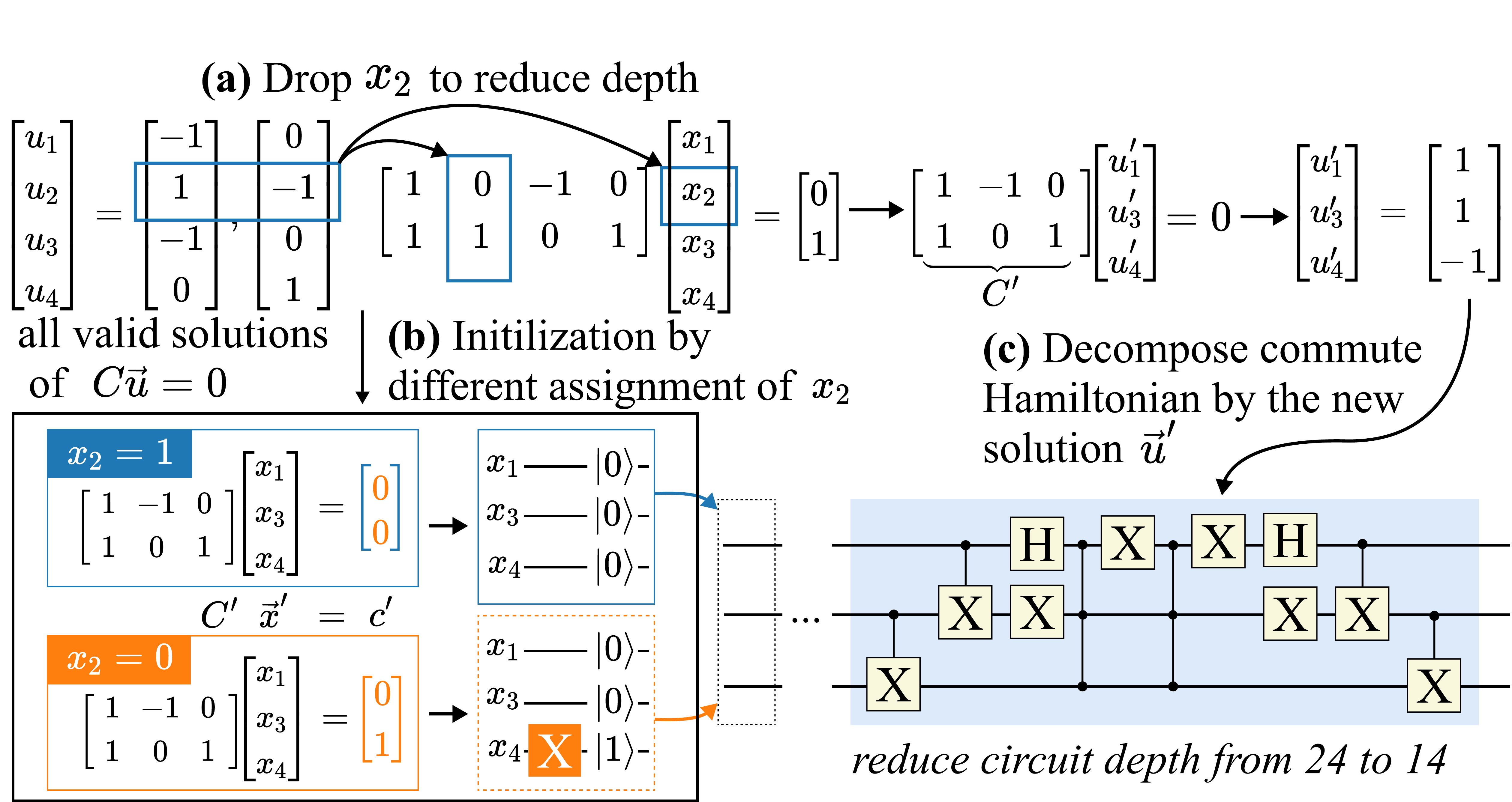}
    \vspace{-0.25cm}
    \caption{Eliminating the variable that involves most non-zero elements across all solutions of $C\vec{u}=0$, which brings higher deployability.}
    \label{fig:binaryhybrid}
    \vspace{-0.35cm}
\end{figure}

Unlike prior approximation-based unitary decomposition~\cite{tan2023quct,bocharov2015efficient,probabilisticunitarysynthesis}, we design a precise decomposition formulation for implementing the commute Hamiltonian with linear complexity. Instead of numerically calculating the Hamiltonian through tensor computations in Equation~\eqref{eq:hd}, we can directly derive the decomposed circuit from the solution vector $\vec{u}$. 
Following the example $H_c(\vec{u}^1)$ in Equation \eqref{eq:hdexample}, we can calculate the eigenstate according to Equation \eqref{eq:x_state}.
\begin{equation}
    \begin{aligned}
        & \vec{u}^{1}=[-1,1,-1,0] \\
    \rm Eq\;(12)\Rightarrow~~~   & |x^{\pm}\rangle = (|010\rangle \pm  |101\rangle)/\sqrt{2} \notag
    \end{aligned}
\end{equation}
Figure~\ref{fig:phasegatedecomp} provides the end-to-end decomposition flow for this Hamiltonian. The $G$ gates are constructed by two CX gates and an H gate. Since $q_2$ and $q_3$ are not in the same state, the first CX gate sets the qubit $q_3$ into $\ket{1}$. Similarly, the second CX gate turns qubit $q_2$ into $\ket{1}$. Then, the H gate on qubit $q_1$ changes the state $(\ket{0}+ \ket{1})\ket{11}/\sqrt{2}$ to $\ket{011}$ and changes the state $(\ket{0}- \ket{1})\ket{11}/\sqrt{2}$ to $\ket{111}$. The phase gate $X_1 P(-\beta) X_1$ puts $e^{-i\beta}$ to state $\ket{011}$, and $P(\beta)$ puts phase $e^{i\beta}$ to state $\ket{111}$. After applying the inversion gate $G^\dagger$, we successfully add phase $e^{i\mp \beta}$ on state $|x^{\pm}\rangle$, constructing the whole circuit for the Hamiltonian $H_c(\vec{u}^1)$. Note that we also draw the resultant circuit of the overall driver Hamiltonian in Figure \ref{fig:seqdriver}, illustrating the linear complexity of circuit depth.

\subsection{Variable Elimination for Higher Deployability}
\label{subsec: divide}

Though our decomposition method has greatly reduced the circuit depth, it is still intolerable for the current NISQ devices. For example, solving a facility location problem with 28 variables features a 989-depth circuit after Hamiltonian decomposition, which goes far beyond the executable depth ($\sim$100) in real-world quantum platforms. Besides, the limited number of qubits also aggravates the difficulty of deployment. An orthogonal optimization direction is to reduce the searching space, i.e., squeeze the constraint matrix by eliminating part of the variables.  

To this end, we further deeply investigate the relationship between circuit depth and the complexity of constraints. In Lemma 2, we have demonstrated that the circuit depth is linear to the number of qubits. Especially, the depth of each local Hamiltonian is linear to the non-zero elements in the solution $\vec{u}$. For example, in Figure \ref{fig:seqdriver} (b), $\vec{u}^1$ and $\vec{u}^2$ involves 3 and 2 non-zero elements, respectively. Thus, the circuit block of Hamiltonian $e^{-i\beta H_c(\vec{u}^1)}$ and $e^{-i\beta H_c(\vec{u}^2)}$ takes 3 and 2 qubits. \emph{In conclusion, the circuit depth of commute Hamiltonian simulation is proportional to the total number of nonzero elements in all solution vectors $\vec{u}\in \Delta$ of $C\vec{u}=0$.}

Based on this observation, we propose a variable elimination technique to reduce the number of variables, which further helps to reduce the dimensions of the solution vector $\vec{u}$. Figure~\ref{fig:binaryhybrid} shows an example that eliminates the variable $x_2$, reducing the solution vector to a triple $\vec{u'}=\{u_1',u_3',u_4'\}$. And the driver Hamiltonian is then derived from the new constraint matrix $C'\vec{u'}=0$. Accordingly, the number of non-zeros is reduced from 5 (3$+$2 in $\vec{u}^1$ and $\vec{u}^2$) to 3 ($\vec{u'}$). Consequently, the circuit depth and the required qubits are cut down, from 24-depth to 14-depth, and from 4 qubits to 3 qubits. However, such source saving introduces the overhead of more measurements, diminishing the quantum advantage. This approach requires enumerating the possible assignment of the eliminated variables and executing the circuit individually. For example, in Figure~\ref{fig:binaryhybrid}~(b), setting $x_2$ to 0 or 1 leads to different initialization when solving {$C'\Vec{x}'=c'$}. And we need to measure these two circuits to construct the final results. 

Theoretically, the outputs after eliminating variables still satisfy the original constraints. The constraint $\sum_{i=1}^{n} c_i x_i = c$ is transformed into $\sum_{i \neq j} c_i x_i = c - c_j x_j$ after eliminating variable $x_j$. Consequently, the commute Hamiltonian is generated from the new constraint $\sum_{i \neq j} c_i u_i = 0$ with the initial state $x^o$ set by the solution of the new constraints $\sum_{i \neq j} c_i x^0_i = c - c_j x_j$. This means $\sum_{i \neq j} c_i x'_i + c_j x_j = c$, which demonstrates that the results (including the evolved $x_{j}$) strictly satisfy the original constraints.

Actually, the measurement overhead increases exponentially to the number of eliminated variables. Such variable elimination is fundamentally an approach to make the circuit more deployable by sacrificing the theoretical parallelism of quantum computing. Therefore, we develop an identification method to find the variable that gives rise to a large reduction in the circuit depth. To be specific, we choose the variable that involves most non-zero elements across all solutions of $C\vec{u}=0$. For example, in Figure~\ref{fig:binaryhybrid}~(a), in the solution of $\vec{u}^1$ and $\vec{u}^2$, $x_2$ features two non-zero assignments, making it an eliminable variable. 

%% file: algorithms/basisconvert.tex
\begin{algorithm}[t]
\small
\caption{Implementation of converting gates $G$} \label{alg:Ggate}
\begin{algorithmic} [1]
\REQUIRE solution $\vec{u}=\{u_1,u_2,\cdots,u_n\}$ of $C\vec{u}=0$.
\ENSURE the circuit to implement $G$ gates in Equation~\eqref{eq:ggate}.
% benchmarking dataset $\{P_i^{bench}\}$ and grouping scheme $\{g_{i,j}\}$ of each iteration $i$
\FOR {$i = 1$ to $n$} 
    \STATE $v_i = (1+u_i) / 2$. \textcolor{blue}{// Equation~\ref{eq:x_state}}
\ENDFOR
\STATE \textcolor{blue}{// turn last n-1 qubits into $\ket{11\cdots1}$}
% \STATE \textcolor{blue}{// iterating from last qubit to first qubit}
\FOR {$i = n$ to $2$} 
    
     % \STATE \textcolor{blue}{// flip qubit $q_i$ when qubit $q_{i-1}$ is in state $\ket{1}$ } \\
    \STATE CX with target qubit $q_i$ and control qubit $q_{i-1}$. \\
    
    \IF{$v_i=v_{i-1}$}
        % \STATE \textcolor{blue}{// if $v_i=v_{i-1}$, flip $q_i$ by X gate.}
        \STATE applying X on qubit $q_i$.\\
        
    \ENDIF
\ENDFOR

\STATE \textcolor{blue}{// current state is $\ket{s^{\pm}}=(\ket{0}\pm \ket{1})\ket{1\cdots1}$}
\STATE applying H on the first qubit.\textcolor{blue}{// $\ket{s^{+(-)}}  \to \ket{0(1)1\cdots1}$} 
\end{algorithmic}
% \vspace{-0.2cm}
\end{algorithm}

%% file: section/6_evaluation.tex
\section{Evaluation}
\label{sec:evaluation}
% \vspace{-0.1cm}
\subsection{Experiment Setup}
\textbf{Benchmark.} 
We evaluate \papername using three practical application scenarios, including facility location problem (FLP)~\cite{melo2009facilitylocation}, graph coloring problem (GCP)~\cite{graphcoloring}, and k-partition problem~\cite{kpartition}. Figure~\ref{fig:casestury} provides detailed problem formulation with their objective functions and constraints. For each application, we collect 400 cases from the related literature~\cite{melo2009facilitylocation,graphcoloring,kpartition}. And we categorize these cases into four problem scales (e.g., F1 to F4 in FLP), with the number of variables ranging from 6 to 28, and the number of constraints ranging from 3 to 16. 

\input{tab/performance}

% We evaluate \papername with 841 problems from the three cases introduced in Section~\ref{sec:casestudy} with different problem scales, which are listed in Table~\ref{tab:speed_noise_free}. Each benchmark is named by its problem scale. 
% For facility location problems, 2F-1D means it has two facilities and one demand. For graph coloring problems, 3V-1E means the graph has three vertexes and one edge. For K-partition problems, 4V-3E-2B means the graph has four vertexes and three edges and the graph will be partitioned into 2 blocks. 
% Their number of variables ranges from 6 to 28, and the number of constraints ranges from 3 to 16.
% Each benchmark contains 100 problems with different settings derived from research works~\cite{kpartition,melo2009facilitylocation}, except for the graph coloring problem~\cite{graphcoloring}, which is limited by the number of edges in a graph.

\textbf{Comparison.} We compare \papername with previous QAOAs that support constrained binary optimization, including Penalty-based QAOA\cite{verma2022penalty}, cyclic Hamiltonian-based QAOA\cite{cyclicdriver}, and Hardware-efficient ansatz (HEA) \cite{HEA}. For Penalty-based QAOA, We integrate it with two state-of-the-art QAOA optimization techniques, \textit{FrozenQubits}~\cite{frozenqubits} and \textit{Red-QAOA}~\cite{redqaoa}. Specifically, \textit{FrozenQubits} aims to reduce the circuit depth and enhance the success rate, while \textit{Red-QAOA} optimizes initial parameters. HEA is a universal variational quantum algorithm (non-QAOA)~\cite{stein2022eqc}, where we use the circuit provided by Kandala et.al~\cite{HEA}. When designing HEA, we introduce a penalty method to make the output satisfy the constraints as much as possible. For parameter updating, we use the constrained optimization by linear approximation method~\cite{cobyla} for all designs.

% Penaly-based QAOA\cite{verma2022penalty}, Cyclic-based QAOA\cite{cyclicdriver}, and Hardware-efficient ansatz\cite{HEA}. For Penaly-based QAOA, We apply FrozenQubits~\cite{frozenqubits} and Red-QAOA~\cite{redqaoa} as baseline. FrozenQubits tends to reduce the circuit depth and upgrade success rate. Red-QAOA is used to find the optimal initial parameter. They both have open-source code, so we integrate their code into the framework of QAOA with the penalty. We combine these two works to improve the success rate of QAOA with the penalty. For the Cyclic Hamiltonian, we use the procedure provided by T. Yoshioka et.al.~\cite{cyclicdriver} as baselines.  Hardware-efficient ansatz is a universal variational quantum algorithm. We use the circuit structure given by Kandala et.al~\cite{HEA} and the objective function is set as Penalty-based QAOA. The penalty coefficient is set by 40 to ensure the output satisfies the constraints as much as possible. 

\textbf{Platform.} 
We conduct several small-scale evaluations on three IBMQ systems, including \textit{Fez} platform with 159-qubit Heron r2 type, \textit{Sherbrooke} and \textit{Osaka} platform with 127-qubit Eagle r3 type~\cite{ibmq}. Since QAOA usually comprises a large amount of CZ gates, \textit{Fez} is QAOA-friendly as it features the CZ gate as the basic gate with 99.7\% fidelity. The other two devices only support single-direction ECR gates with 99.3\% fidelity, which takes three ECR gates to implement the CZ gate, resulting in a higher error rate. The simulation experiments and the classical part of QAOA are executed on an AMD EPYC 9554 64-core sever with 1.5T SSD memory. The simulation of the quantum circuit is accelerated by one A100 GPU on the server.

% For the performance comparison, we evaluate Choco-Q on three IBMQ systems, including Fez with Heron r2 type and 156 qubits, Sherbrooke and Osaka with Eagle r3 type and 127 qubits. Fez has the highest Fidelity (CZ gate fidelity is 99.7\%) and supports the CZ gate as the elemental gate, which is friendly for the QAOA algorithms. The other two devices only support single-direction ECR gates with fidelity of 99.3\% and one CZ gate needs 3 ECR gates to implement. Thereby, they have a higher error rate when deploying the QAOA algorithms. We select the devices with different fidelity levels to analyze the effect of noise.

% \textbf{Choco-Q Implementation.} 
% Choco-Q is implemented by Python. We use Python to implement all the techniques mentioned in the paper. Qiskit package is used to build and execute the circuit on the simulator or the IBM quantum cloud machine. For parameter updating, we use the Constrained Optimization By Linear Approximation (COBYLA) method for both baselines and Choco-Q for fair comparison. simulation.

\textbf{Evaluation metrics.} 
Similar to prior constraint QAOAs, we employ three algorithmic metrics: \textit{success rate}, \textit{in-constraints rate}, and \textit{approximation ratio gap}. The \textit{success rate} is defined as the probability of getting the optimal solution after measurements. The \textit{in-constraints rate} is the probability that the output solutions satisfy the constraints. Thus, the in-constraints rate is always higher than the success rate. \textit{Approximation ratio gap} (ARG) indexes the solution quality across all outcomes generated by the algorithm, which is widely used in prior QAOAs \cite{frozenqubits,redqaoa,groveradaptivesearch}. It is defined as
\begin{equation}
    ARG =|\frac{E(f(\vec{x}) +\lambda\|C\vec{x}-\vec{c}\|)}{f(\vec{x_{optimal}})}-1|  \label{eq:arg}
\end{equation}
Here, the objective function $f$ and constraints $C$ are defined in Equation~\eqref {eq:cboexample}. $\vec{x}$ is the outcome, $\vec{x_{optimal}}$ is the optimal solution, and $E$ means the expectation value on all outcomes. $\lambda$ is the penalty term, which is set to 10 here to balance the objective and constraints.
When implementing on real-world devices, we further evaluate the end-to-end \textit{latency} (without data communication), including the compilation time for Hamiltonian decomposition, circuit execution time, and the parameter updating time for iterative optimization.

\vspace{-0.1cm}
% \subsection{Choco-Q Performance Overview}

\subsection{Algorithmic Evaluation}
% In this section, we evaluate Choco-Q on the noise-free simulator for theoretical performance evaluation, and the results are shown in Table~\ref{tab:speed_noise_free}. Since QAOA algorithms need multiple layers to evolve to the optimal solution, we set the number of layers of baselines by seven to ensure their evaluation. 
% For Choco-Q, we only use one layer since the commute Hamiltonian has been serialized to include all search directions.

\textbf{Success rate.} 
Compared to other QAOA designs, \papername achieves an average 235$\times$ improvement in successfully finding the optimal solution, as shown in Table~\ref{tab:speed_noise_free}. In particular, for the problem in medium-scale (F2, G2, K2), all other approaches nearly fail to get the optimal solution, with a success rate below 15\%. While, our approach exhibits a high success rate, from 52.6\% to 67.1\%. Even for large-scale problems, we still have a relatively high chance of finding the solution (9.50\% - 30.0\%). Such great improvement comes from the fact that the commute Hamiltonian effectively narrows down the search space, leading to a higher probability of getting the optimal solution.
% Among different benchmarks, \papername performs excellently in the graph coloring problem, e.g., 99.9\% and 99.8\% in G1 and G2 cases, respectively. \hl{This is because GCP has more constraints than other QAOA designs in the same problem scale. The search space in GCP is thus smaller, resulting in a higher success rate in this benchmark.}

\textbf{In-constraints rate.}
When encoding the constraints, \papername shows a 100\% in-constraints rate, indicating that our algorithm comprehensively takes all constraints into account. This is also the reason that \papername has a much higher success rate. The 100\% in-constraints rate benefits from the generality of our encoding scheme that applies commute Hamiltonian to thoroughly represent the constraints. In contrast, the cyclic Hamiltonian restricts constraints to a summation format, making it challenging to consider all constraints, especially in large-scale problems. Besides, the cyclic Hamiltonian performs better on KPP benchmarks since the constraints of KPP are in summation format and involve fewer shared variables.
The in-constraints rate of penalty-based QAOA~\cite{verma2022penalty} is highly determined by the number of constraints, where each constraint equation corresponds to one penalty term in the objective function. In contrast, \papername is irrelevant to it as we consider the problem as the commute operator between the driver Hamiltonian and the constraints. 

\textbf{Approximation ratio gap (ARG).}  As shown in Table~\ref{tab:speed_noise_free}, Choco-Q improves the ARG by 658$\times$. Such improvement originates from that in Choco-Q, $\lambda|C\vec{x}-\vec{c}|$ in Equation~\eqref{eq:arg} is constantly zero, while it can be huge using other methods. In particular, the ARG is larger in GCP benchmarks since GCP contains more complicated constraints. Even though, the ARG of Choco-Q is below 0.6.

\textbf{Circuit depth.}
\papername involves a little bit more circuit depth, notably in large-scale problems. First, as mentioned in Section \ref{subsec:phasegate}, prior algorithms are fundamentally approximation-based approaches for optimization problems, which may easily fail to meet the constraints. However, \papername lies on the precise encoding of arbitrary linear constraints, which requires implementing the commute Hamiltonian of all solution $\vec{u}$ in Equation \ref{eq:hd}, resulting in the linear complexity of circuit depth. Taking the G3 case as an example, it results in 12 $\vec{u}$ to precisely express the 12 constraint equations, giving rise to large circuit depth. The second reason is that we simulate the other QAOA algorithm only with seven repeated layers, i.e., repeat 7 times in Equation \eqref{eq:circuit}. We adopt seven layers because this setting achieves the best trade-off between success rate and circuit depth. We test that even if we increase the number of layers, it only contributes to limited improvement in the success rate. And in Table~\ref{tab:speed_noise_free}, we do not repeat the driver Hamiltonian of \papername (one layer). In a word, compared to the significant performance improvement, the additional circuit depth is acceptable.

% We can see the trend of ARG on benchmark scales does not align with the success rate, especially for SCP ben 

% As shown in Table~\ref{tab:speed_noise_free}, Choco-Q overcomes baselines in every benchmarks. In benchmark with small variables, Choco-Q improves the success rate from around 10\% to at least 56\%. Although the success rate decays with the problem scale, Choco-Q remains over 1\% success rate on the benchmark with over 20 variables while baselines all approach zero. For the in-constraints rate,Choco-Q can achieve 100\% in any scale of benchmarks while baselines decay with the problem scale. We observed that in baselines, Cyclic-based QAOA has the best success rate with 2.5\% higher than others. This indicates encoding the constraint into driver Hamiltonian is a better choice than encoding into the objective. However, the in-constraints rate of Cyclic-based QAOA is 2.12\% since it only supports the summation format of constraints. The in-constraints rate of Penalty-based QAOA is 35.51\%, which is the highest in baselines. The value of the penalty term is quite large so Penalty-based QAOA concentrates on satisfying the constraints while ignoring the optimization of the target objective. The total improvement in Figure~\ref{fig: converge}~(a). The improvement is calculated by $\frac{\text{Choco-Q}}{\text{baselines}}$. Overall, Choco-Q improves baselines by $10^3\times$ at most and $7.8\times$ at least.
\begin{figure}[t]
    \centering
    \includegraphics[width=.99\linewidth]{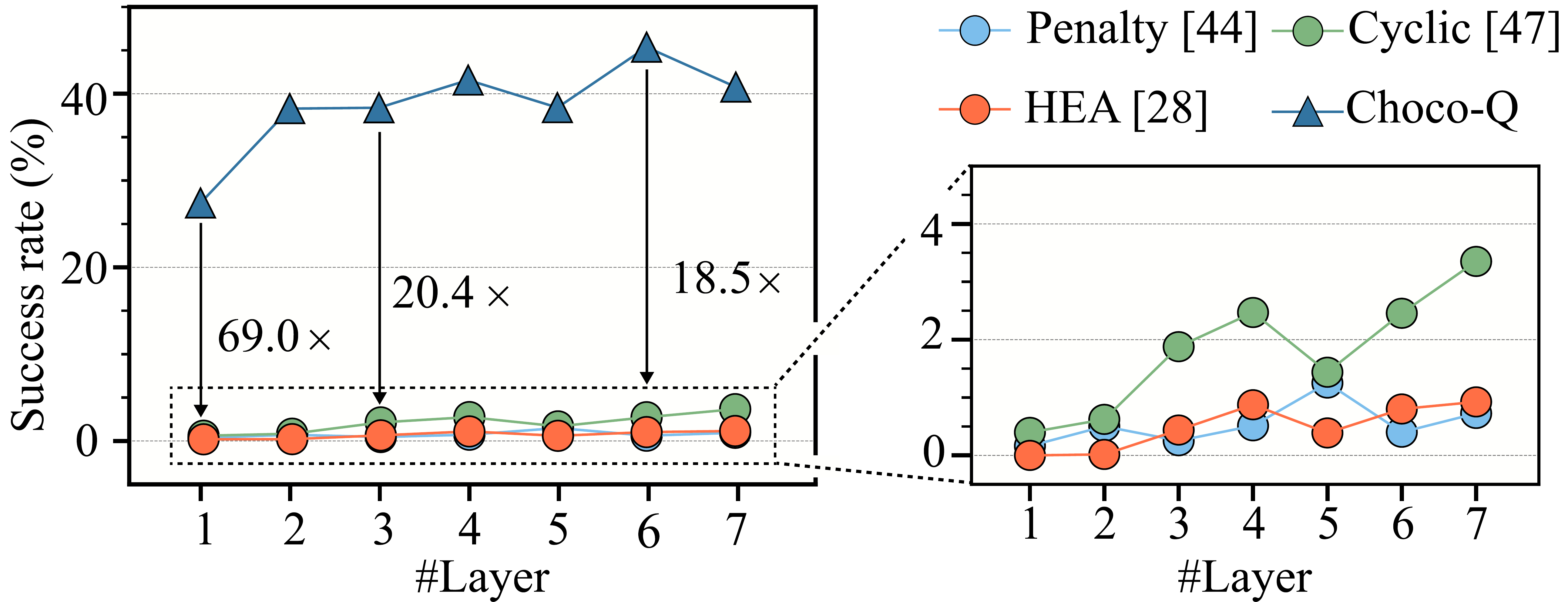}
    \vspace{-0.8cm}
    \caption{The average success rate using the different number of layers.}
    \label{fig:layer}
\end{figure}

\begin{figure}[t]
    \centering
    \includegraphics[width=.90\linewidth]{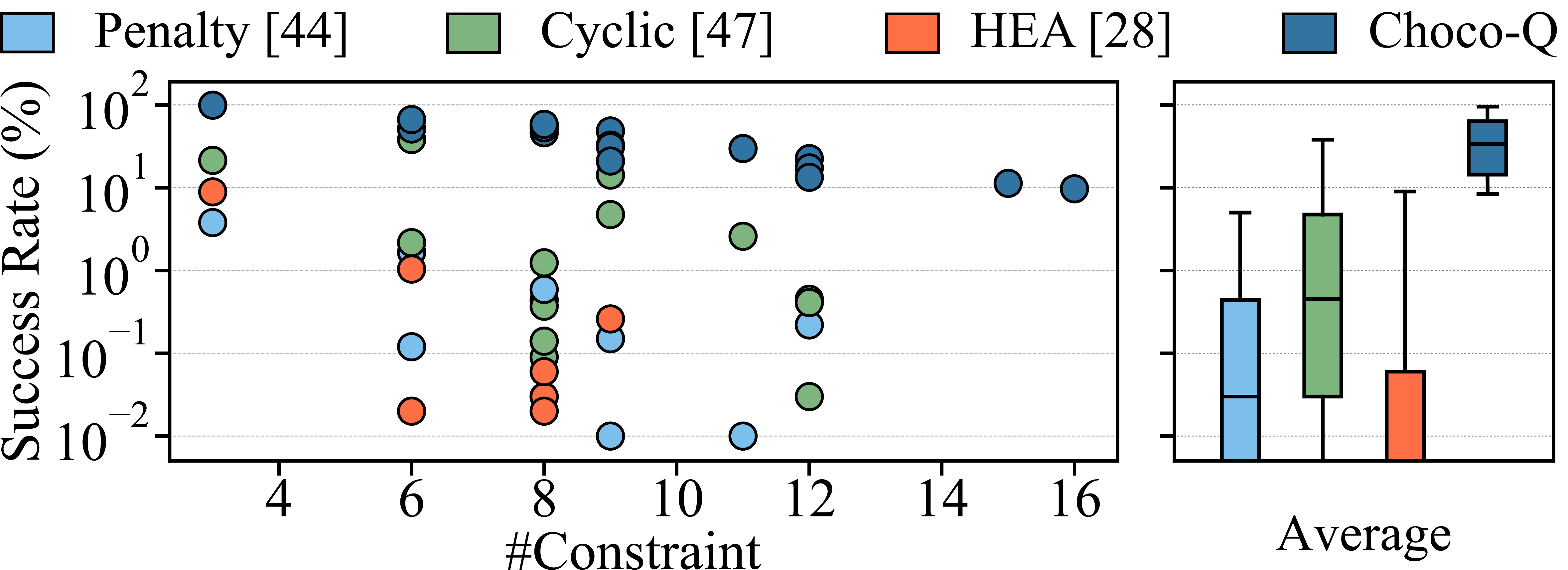}
    \vspace{-0.3cm}
    \caption{The success rate and circuit depth of Choco-Q under the different number of constraints.}
    \label{fig:cons}
    \vspace{-0.15cm}
\end{figure}

\begin{figure}[t]
    \centering
    \includegraphics[width=0.575\linewidth]{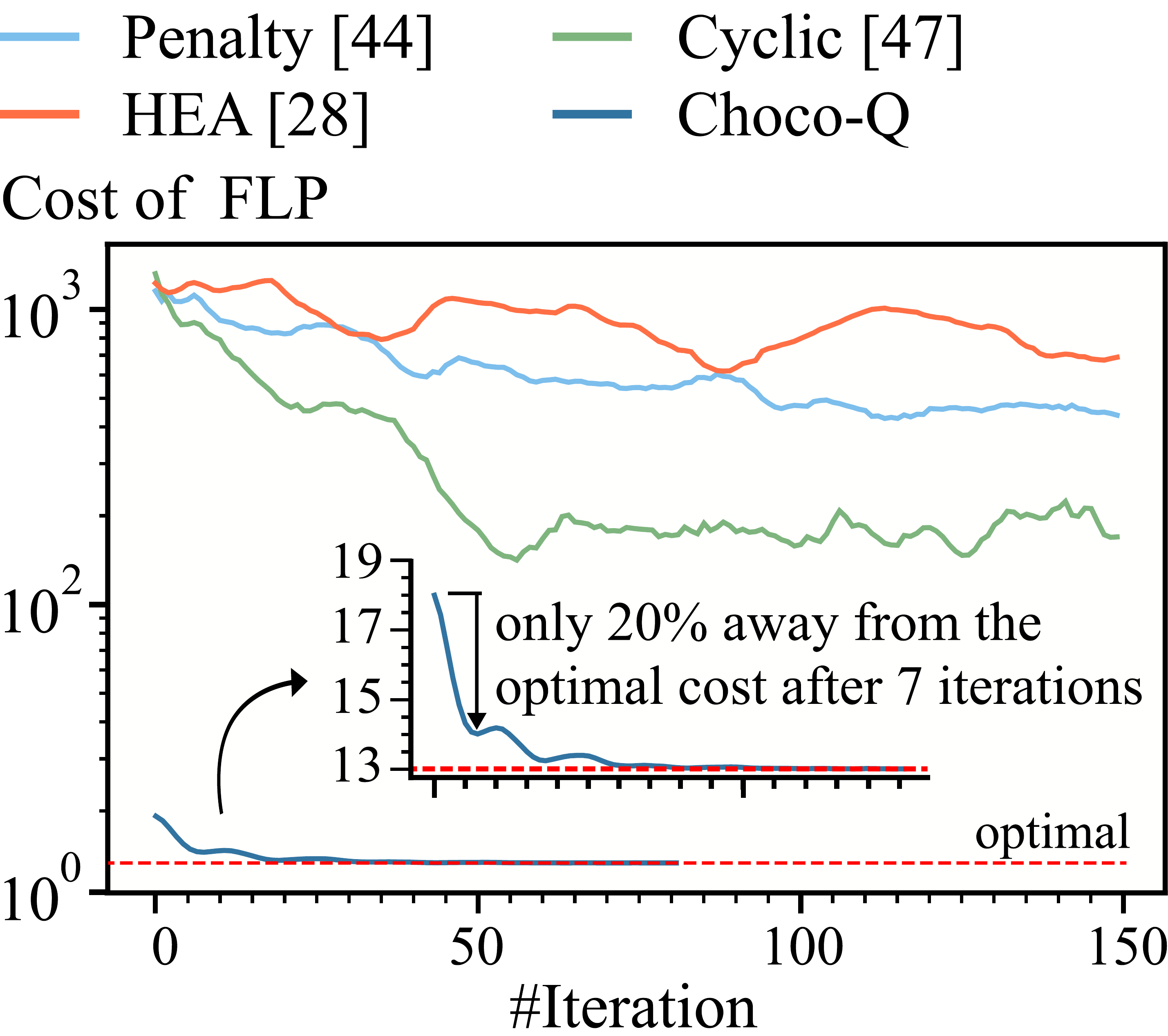}
    \hfill
    \includegraphics[width=0.38\linewidth]{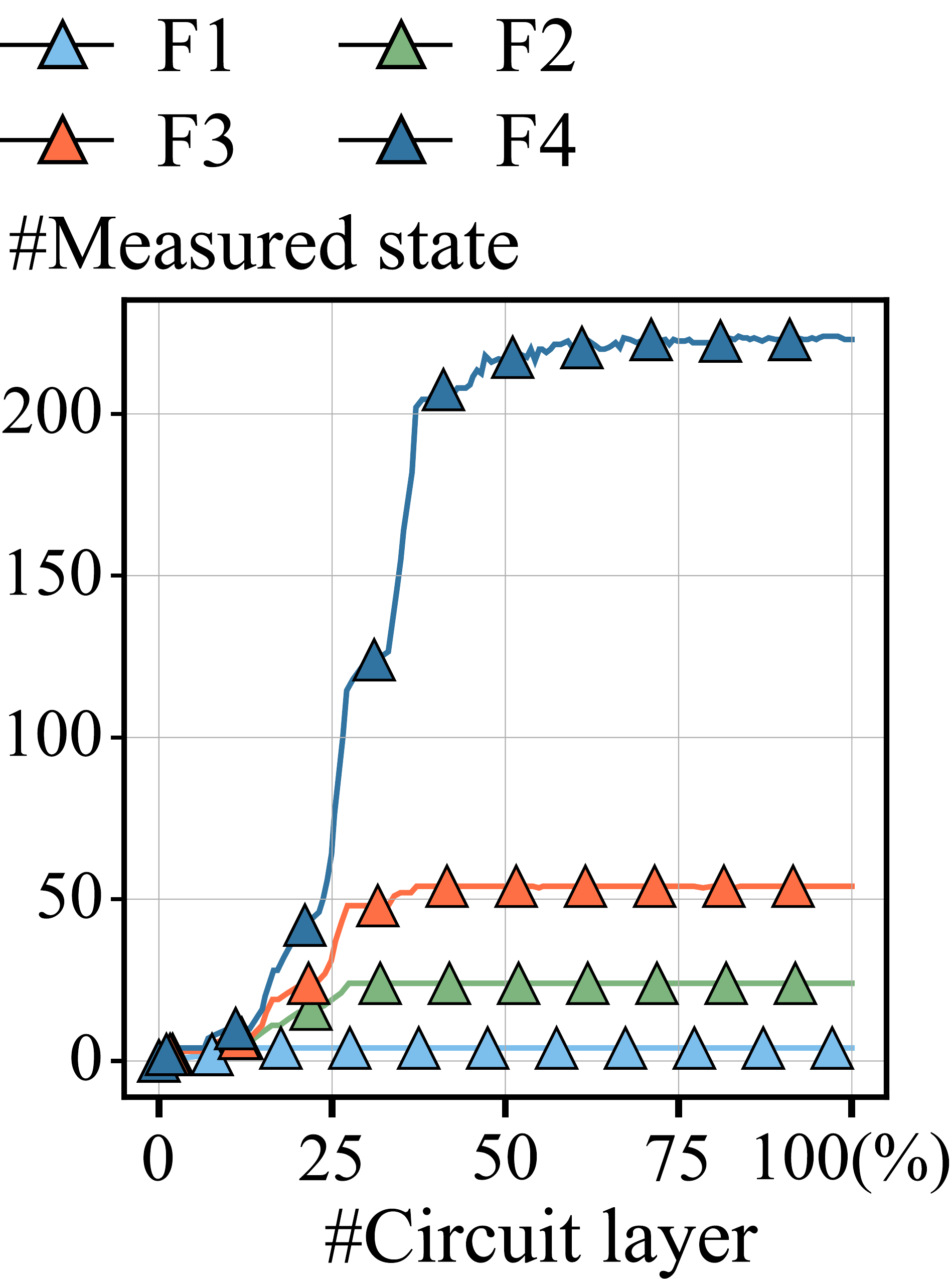}
     \\
      \vspace{-0.1cm}
    \makebox[0.54\linewidth][c]{\parbox{0.54\linewidth}{\footnotesize \centering (a) Convergence using F1:2F-1D.}}\hfill
    \makebox[0.42\linewidth][c]{\parbox{0.42\linewidth}{\footnotesize (b) The parallelism of Choco-Q.}}
    \\
    \vspace{-.25cm} 
    \caption{Convergence analysis of Choco-Q.}
    \label{fig: converge}
\end{figure}

\textbf{Success rate with the different number of repeated layers.} Figure~\ref{fig:layer} evaluates the average success rate by varying the number of repeated layers. We can see that the success rate of our design is always over 25\% while the others are less than 5\%. In \papername, applying two repeated commute Hamiltonian layers improves the success rate from 27.4\% to 38.3\%. Since the serialization of commute Hamiltonian has covered all searching directions, more layers can only bring limited improvement. On the contrary, increasing the number of layers significantly improves the success rate of other designs. For example, in the cyclic Hamiltonian-based QAOA~\cite{cyclicdriver}, each additional layer leads to an average of 0.5\% improvement. Even though, the final success rate is still hindered by the algorithmic limitations. 

\textbf{Success rate with the different number of constraints.}  
Figure~\ref{fig:cons} plots the success rate of all graph benchmarks with the number of constraints on the x-axis. The advantage of Choco-Q becomes more significant with the growing number of constraints. In particular, when the number of constraints exceeds 12, the success rate of other methods almost becomes zero, while Choco-Q still holds a success rate of more than 10\%. This achievement arises because the commute Hamiltonian strictly limits the search space under constraints.

\textbf{Convergence.}
Figure~\ref{fig: converge}~(a) depicts the convergence curves of different QAOA designs, picking one case from the F1:2F-1D benchmark. 
Choco-Q can reach the optimal cost within 30 iterations, while other methods require over 148 iterations and still have at least a 78\% gap with the optimal cost. In particular, Choco-Q can quickly reach 20\% away from the optimal cost within 7 iterations, which is $7.3\times$ faster than the cyclic Hamiltonian-based approach. Such high convergence speedup is attributed to a good initial cost (18), while the others are extremely large ($10^3$). Furthermore, the squeezed search space also helps to find the optimal solution in fewer iterations. Figure~\ref{fig: converge}~(b) illustrates the parallelism of Choco-Q. We collect the number of measured states through the circuit, which represents the parallelism involved in the superposition of quantum states. Choco-Q shows an exponential growth of the parallelism at the beginning of the circuit, i.e., at the point of around 1/4 circuit. Unlike prior QAOAs that initialize a uniform superposition state, even though Choco-Q prepares a special initial state, we still effectively harvest the quantum parallelism by applying commute Hamiltonian.

\subsection{Evaluation on Real-world Quantum Platforms.}
This section evaluates Choco-Q on the IBM cloud quantum devices. Considering the short decoherence time and the inevitable noise error, we choose to implement the small-scale problems, i.e., F1, G1, and K1 cases in Table~\ref{tab:speed_noise_free}.

\begin{figure}[t]
    \centering
    \includegraphics[width=.99\linewidth]{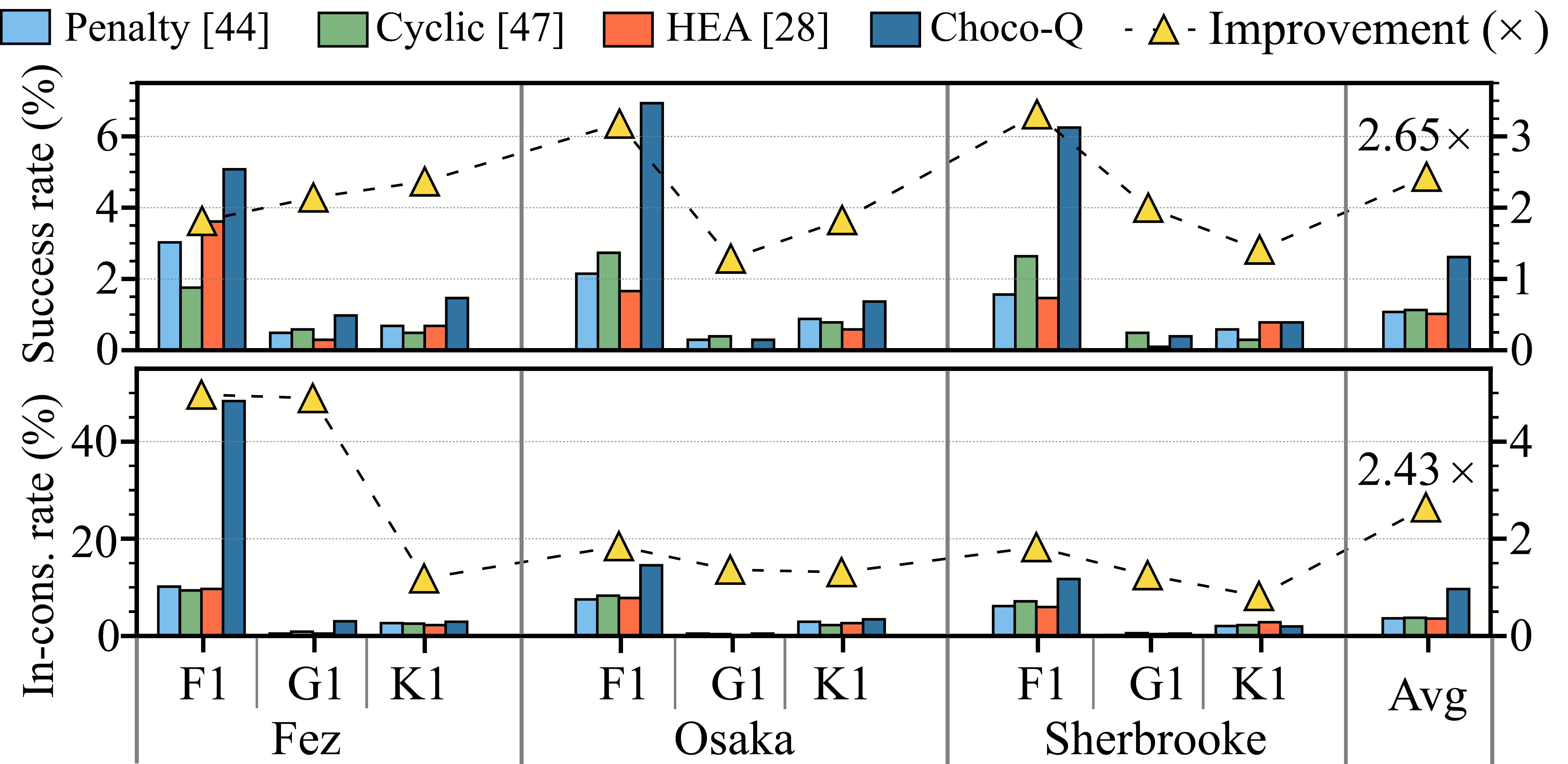}
    \vspace{-0.7cm}
    \caption{Success rate and in-constraints rate on IBM quantum platforms.}
    \label{fig:realqc}
\end{figure}

\textbf{Success rate.} Figure~\ref{fig:realqc} gives the success rate on three quantum devices. Compared to the noise-free simulator, the success rate of all cases is decreased due to various noise, e.g., cross-talk and hardware defects. Overall, Choco-Q achieves an average of $2.65\times$ improvement over other designs. 
The effectiveness of NISQ devices arises from our decomposition technique that transforms the commute Hamiltonian into multiple phase gates, which are less sensitive to flip error. 
We observe that all methods on the G1 benchmark have the lowest success rate. This is because G1 requires 12 qubits to solve and involves more cross-talk errors. The circuits of F1 have a higher success rate since they only consist of six variables and three constraints.

% The improvement is averaged over three baselines. Overall, Choco-Q achieves average $2.65\times$ and $2.43\times$ improvement on success rate and in-constraints rate, respectively. In particular, the success rate of Choco-Q is over 5\% in the F1 benchmark. We observe that all method on the G1 benchmark has the lowest success rate. This is due to G1:3V-1E needing 12 qubits to solve and the cross-talk noise between qubits causes much error. Even in this circumstance, Choco-Q still has a better performance over others. In the Fez device with the highest fidelity, Choco-Q achieves $4\times$ improvement. The reason behind such a performance is that we use a commute Hamiltonian decomposition with phase gates and the qubit flip error does propagate through the phase gates. This decomposition thus improves the robustness in a noisy environment.

\textbf{In-constraints rate.} Compared to other designs, Choco-Q shows $2.43\times$ higher in-constraints rate. In particular, on the Fez platform, we achieve up to 48\% in-constraints rate, resulting in $4.98\times$ improvement, thanks to its architectural properties. As mentioned before, this device is QAOA-friendly, allowing us to fully leverage our technological advantages. As we serialize the driver Hamiltonian into a set of smaller ones that necessitates fewer qubit connections, Choco-Q is more noise-tolerant against the cross-talk noise.

% Figure~\ref{fig:realqc}~(b) shows the results of in-constraints rate. Although noise can make the commute Hamiltonian not commute with the constraints operator, Choco-Q still has an average $2.43\times$ improvement. The in-constraints rate can reach 48\% in Fez devices on the F1 benchmark, which means almost half of the constraints are satisfied. This robust performance originates from our serialization of commute Hamiltonian, which separates the total Hamiltonian into independent local Hamiltonians. The local Hamiltonian contains interactions on fewer qubits and thus the cross-talk noise between qubits is suppressed. 

\begin{figure}[t]
    \centering
    \includegraphics[width=0.95\linewidth]{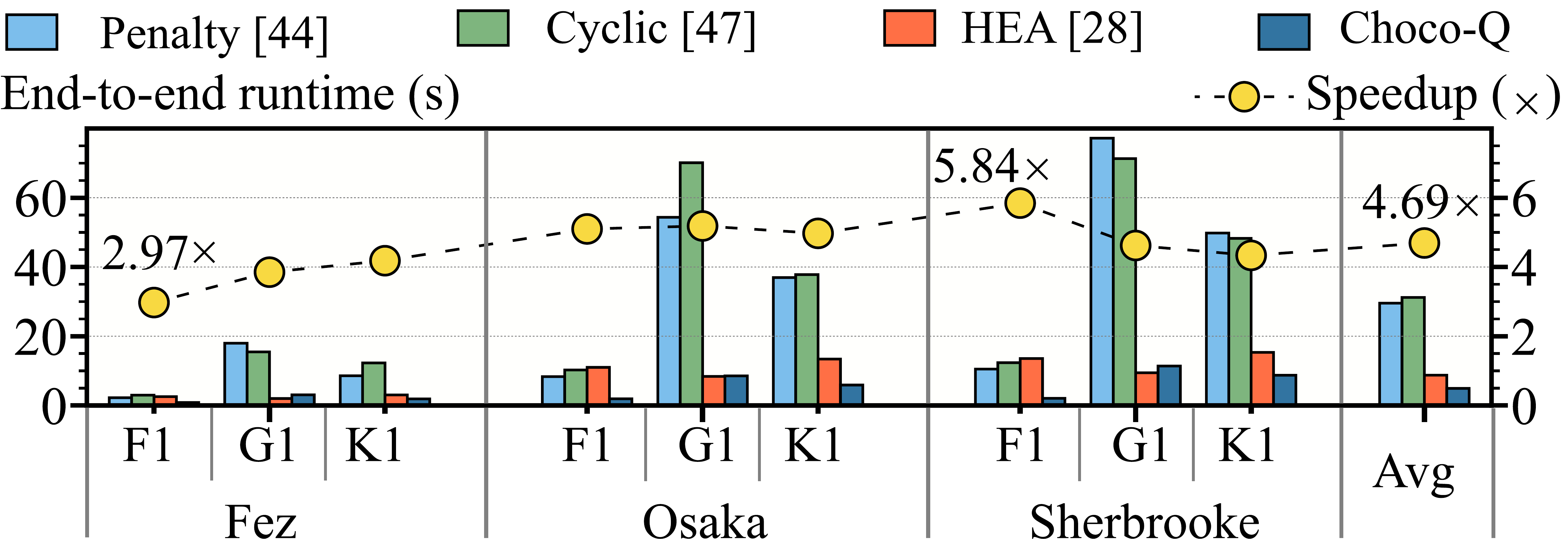}
     \\
     \vspace{-0.1cm} 
    \makebox[0.95\linewidth]{\footnotesize (a) End-to-end latency comparison on the real-world quantum platforms.}
    \vspace{0.15cm} 
    \\
    \includegraphics[width=.98\linewidth]{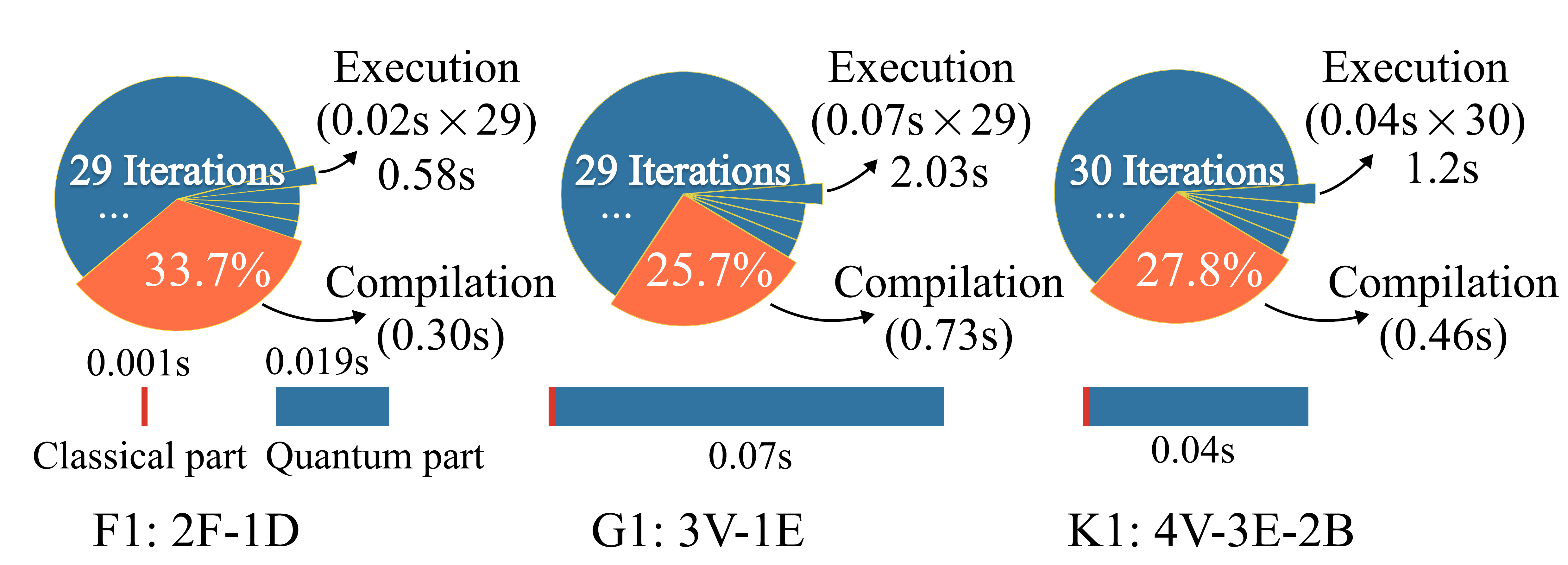}
    \\
    \vspace{-0.1cm} 
    \makebox[0.95\linewidth]{\footnotesize (b) Latency breakdown of Choco-Q on the Fez platform.}
    \\
    \vspace{-0.2cm} 
    \caption{End-to-end latency evaluation and breakdown.}
    \label{fig:runtime}
      \vspace{-0.2cm}
\end{figure}
\textbf{End-to-end latency.} 
Figure~\ref{fig:runtime}~(a) compares the overall latency comprising of the compilation time, circuit execution time, and parameter updating time. Choco-Q achieves $2.97\times$ - $5.84\times$ speedup, and always within 10 seconds, primarily attributed to the reduced number of iterations. The latency of HEA~\cite{HEA}, the non-QAOA algorithm, is close to ours in several cases. This is because it shows shallow circuit depth in F1, G1, and K1 cases, leading to around $4\times$ speedup when executing the circuit. 

Figure~\ref{fig:runtime}~(b) presents the latency breakdown for the cases in the Fez platform. We consider the end-to-end latency as the summation of compilation time (including the decomposition) and execution time. The execution time includes the classical and quantum parts of QAOA, with the classical part only taking a little time. We can see that the iterative execution is the most time-consuming, requiring around 30 iterations and accounting for around 70\% of the total latency.

\subsection{Detailed Analysis of Optimization Techniques} \label{subsec:opteval}
% This section evaluated the optimization techniques used in Choco-Q, including serialization and decomposition of commute Hamiltonian and variable elimination. For commute Hamiltonian decomposition, we randomly generated commute Hamiltonians on different qubits as benchmarks since direct unitary decomposition cannot deal with real-world benchmarks.
% For variable elimination, we selected F2, G2, and K2 as benchmarks and evaluated the results on the three IBM quantum devices.

% \textbf{Evaluation of commute Hamiltonian serialization and decomposition}. 
\textbf{Decomposition of commute Hamiltonian.}
Figure~\ref{fig:decompose} compares the Hamiltonian decomposition results between Trotter decomposition~\cite{trotterbasedsimulation} and Choco-Q. Figure~\ref{fig:decompose}~(a) shows the decomposition time and memory usage across different circuit sizes. By putting our serialization technique and decomposition formulation together, our approach demonstrates a linear time complexity (less than 0.1s) and constant memory usage (less than 10MB). Taking the 10-qubit Hamiltonian as an example, we achieve $10^6\times$, $8341\times$ reduction for the decomposition time and memory usage, respectively. When the number of qubits is beyond 10, the Trotter decomposition fails to output the results. The high scalability mainly stems from that we effectively eliminate the massive tensor computation in Equation \eqref{eq:hd} by serialization and equivalent transformation.

\begin{figure}[t]
    \centering
    \includegraphics[width=0.5225\linewidth]{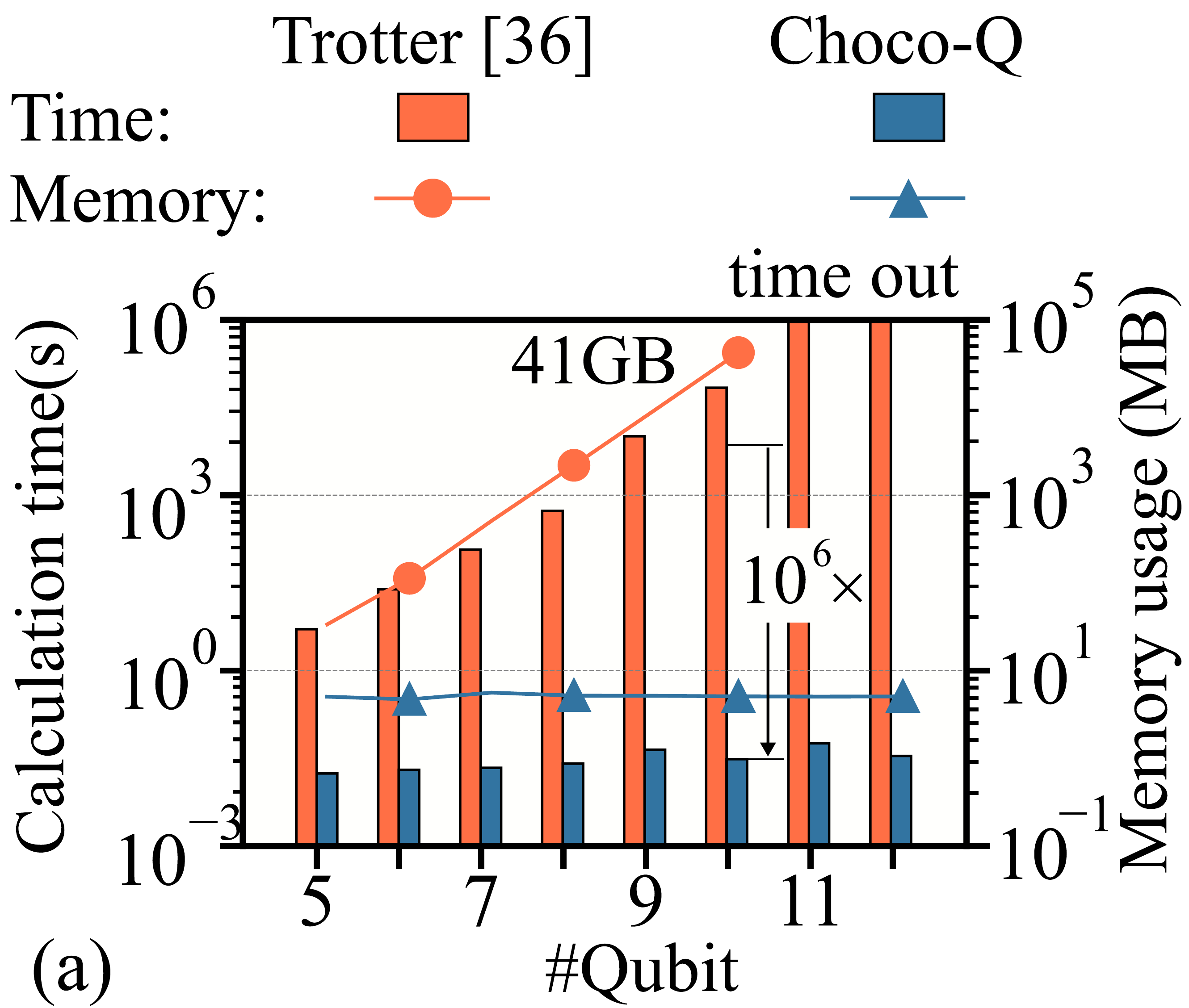}\hfill
    \includegraphics[width=0.4275\linewidth]{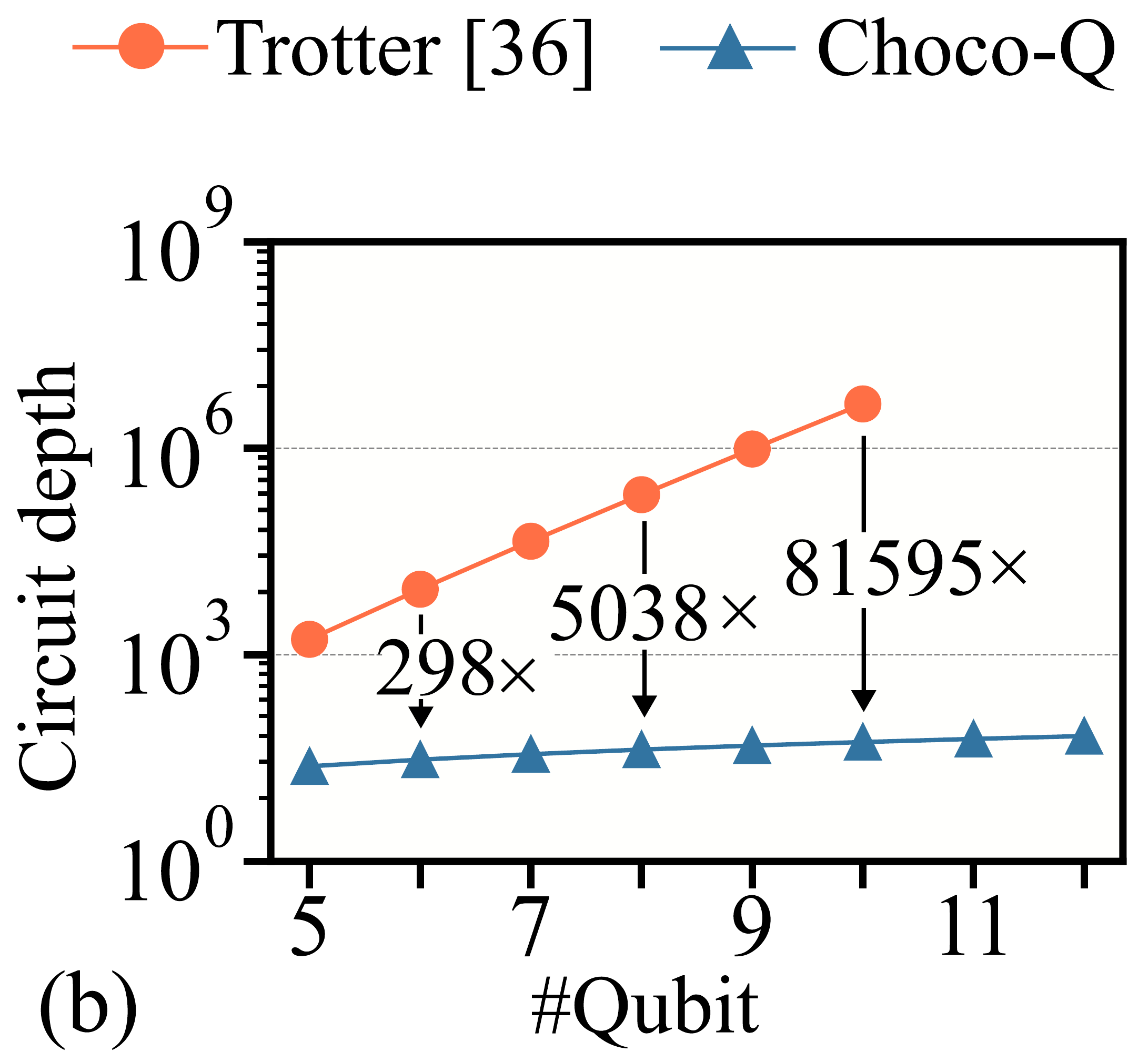}
     \\
    % \makebox[0.6\linewidth][c]{\parbox{0.6\linewidth}{\footnotesize (a) Calculation time and memory\\ usage.}}
    % \makebox[0.3\linewidth][c]{\parbox{0.3\linewidth}{\footnotesize (b) Circuit depth.}}
    % \\
    \vspace{-0.3cm}
    \caption{Comparison between Trotter decomposition~\cite{trotterbasedsimulation} and Choco-Q.}
    \label{fig:decompose}
    % \vspace{-0.2cm}
\end{figure}
% gives the results of trotter decomposition~\cite{trotterbasedsimulation} and Choco-Q when dealing with different scales of commute Hamiltonian.  Figure~\ref{fig:decompose}~(a) gives the calculation time and memory usage when decomposing the commute Hamiltonian. Overall, Choco-Q shows a linear time complexity and constant memory usage. In 10 qubits, Choco-Q reduces the calculation time by $10^6\times$ and reduces the memory usage from 87GB to 7MB. When the number of qubits exceeds 10, the decomposition time of trotter decomposition will exceed $3$ months. The calculation time of Choco-Q is always below 0.1s and the memory usage remains below 10MB. This high speed and low memory is due to that Choco-Q does not handle the unitary by tensor operation, reducing the storage requirements. Besides, Choco-Q serialized the commute Hamiltonian into local Hamiltonian, which can be stored and calculated efficiently. 

Figure~\ref{fig:decompose}~(b) compares the resultant circuit depth using these two decomposition methods. The circuit depth of Choco-Q is linearly increased with the number of qubits, from 24 (5-qubit) to 66 (12-qubit). Such quantum resource saving benefits from our serialization of commute Hamiltonian, which avoids enormous repetition of approximation Hamiltonian. In addition, our decomposition flow in Figure \ref{fig:phasegatedecomp} also ensures that the convert gate and the phase gate can be efficiently decomposed into basic gates.

% gives the circuit depth of the two approaches. The circuit depth of Choco-Q shows a linear increase with qubits while the trotter decomposition shows an exponential growth. In 10 qubits commute Hamiltonian, Choco-Q reduces the circuit depth from $5.7\times10^7$ to $698$, leading to $81595\times$ reduction. The circuit depth reduction benefits from two reasons. one is our serialization of commute Hamiltonian, which avoids enormous repetition of commute Hamiltonian like Equation~\ref{eq:trotter} of trotter decomposition. Another reason is that we use an exact decomposition method to decompose each local Hamiltonian, which has linear complexity as illustrated in Section~\ref{subsec:phasegate}.

\begin{figure}[t]
    \centering
    \includegraphics[width=0.51\linewidth]{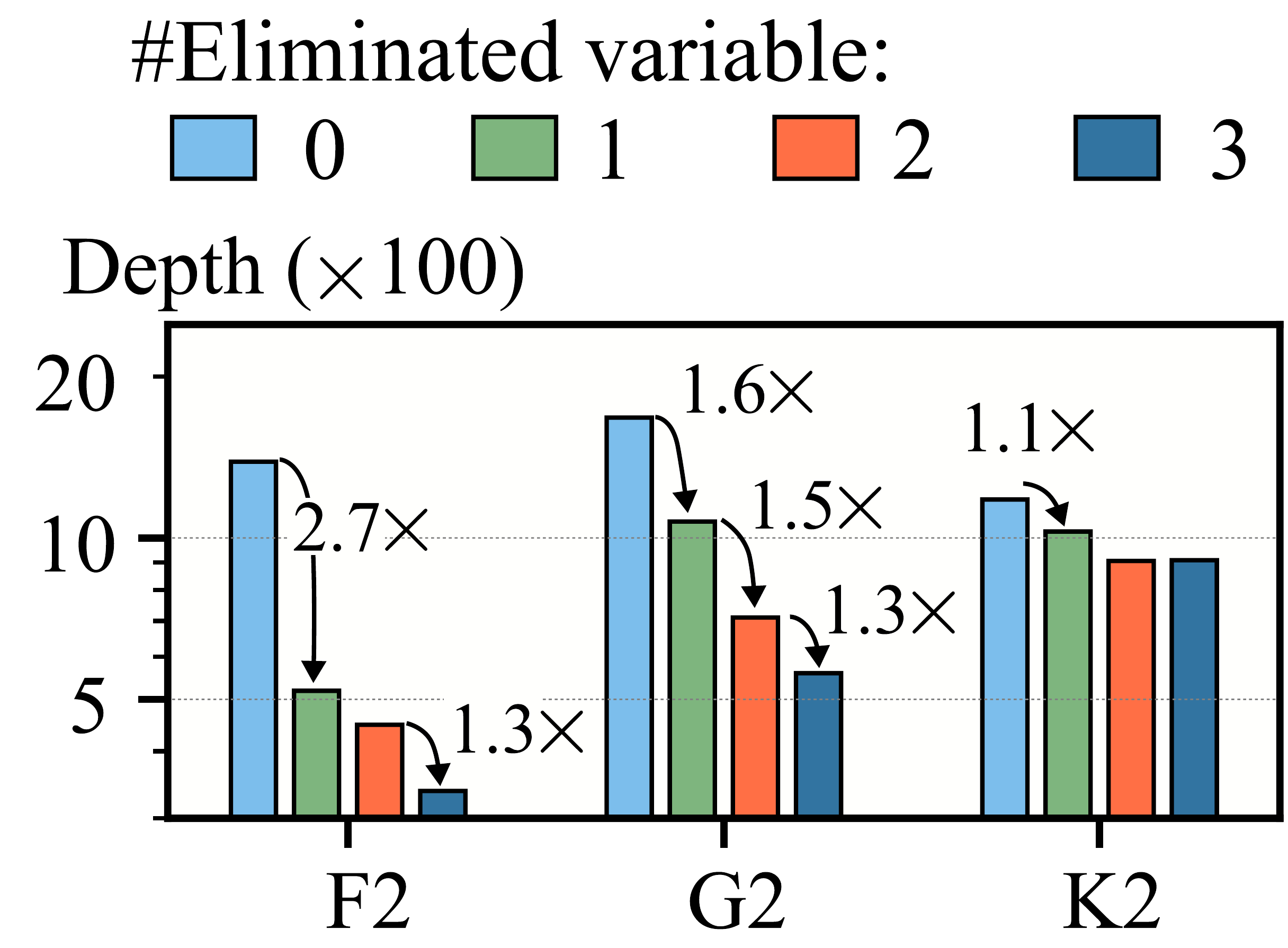}\hfill
    \includegraphics[width=0.44\linewidth]{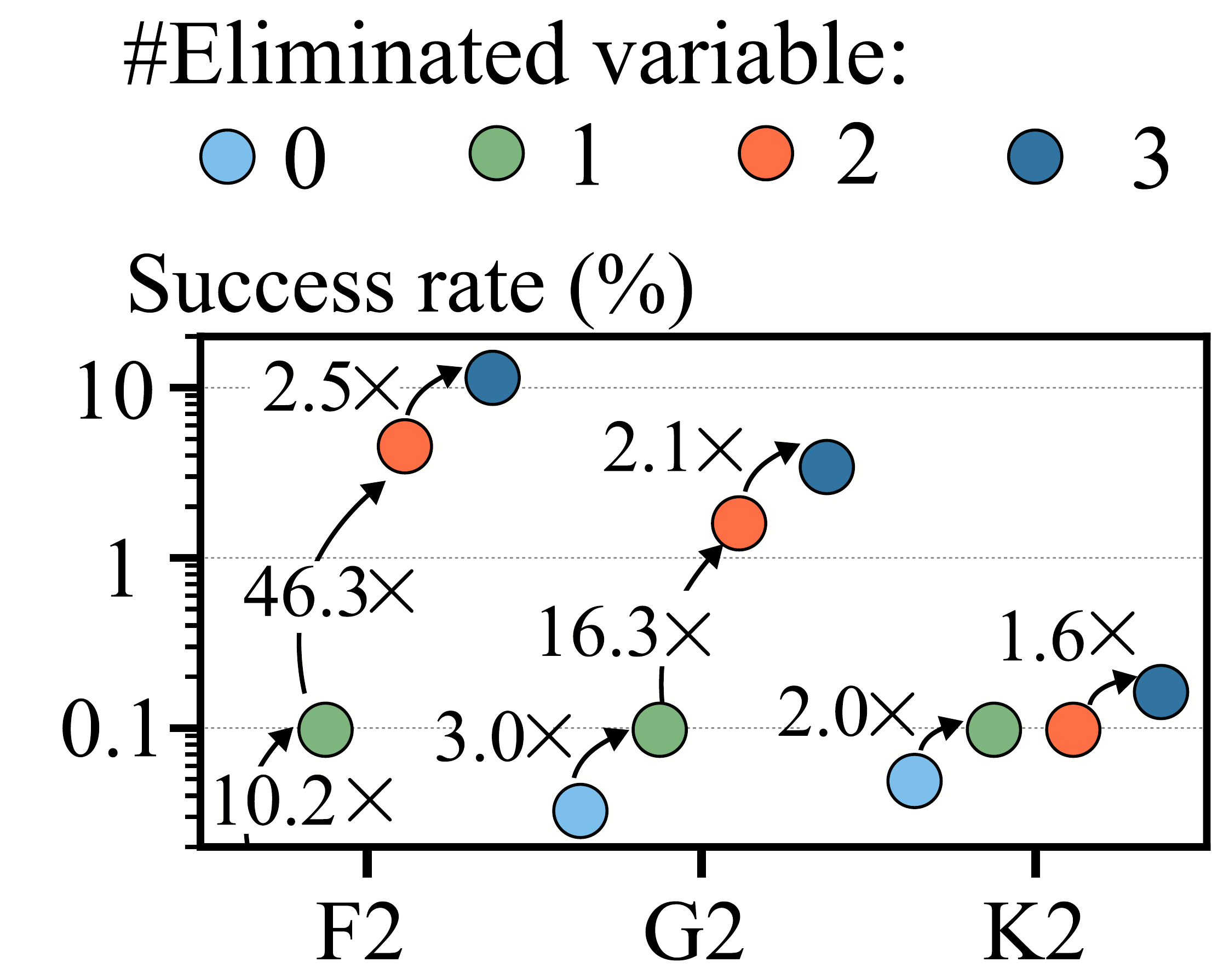}
    \\
    \makebox[0.54\linewidth][c]{\parbox{0.54\linewidth}{\footnotesize  (a) Depth of circuit after eliminating
    the different number of variables.}}\hfill
    \makebox[0.40\linewidth][c]{\parbox{0.40\linewidth}{\footnotesize (b) Success rate improvement\\ under IBMQ noise model.}}
    \\
    \vspace{-0.1cm}
    \caption{Evaluation of variable elimination technique.}
    \label{fig:eliminate}
    \vspace{-0.2cm}
\end{figure}

\textbf{Variable elimination.} We analyze the circuit depth and the success rate after eliminating a different number of variables, as shown in Figure~\ref{fig:eliminate}. Overall, variable elimination produces a remarkable boost in both circuit depth and success rate. For example, in F2 case, just eliminating one variable leads to a 2.7$\times$ circuit depth reduction and yields 10.2$\times$ success rate improvement, which exceeds the overhead caused by extra measurements. We also find that such yields slow down when more eliminated variables are introduced. This is because most of the non-zero values in solutions of $C\vec{u}=0$ have been eliminated and further elimination leads to less benefit. Besides, this technique has little gain in KPP cases due to the uniform distribution of qubits in the local Hamiltonian. Thus, as shown in Figure \ref{fig:binaryhybrid}, dropping one of them cannot greatly reduce the number of non-zeros in all solutions of $C\vec{u}=0$. 

% Different colors represent the number of eliminated variables. zero means not using variable elimination. Overall, the effect of variable elimination is remarkable. In Figure~\ref{fig:eliminate}~(a), when eliminating one variable, the circuit depth is reduced by $1.1\sim2.7\times$. For some benchmarks like K2, the depth reduction is only $1.1\times$. This is due to that the KPP problem has fewer constraints than other problems (Table~\ref{tab:speed_noise_free}), and the distribution of qubits in the local Hamiltonian is relatively uniform, dropping one of them cannot reduce the total depth heavily. 

Moreover, to demonstrate the advantage of variable elimination, we investigate the success rate under the noise model of the three IBM quantum platforms. We can see that the success rate benefits most when eliminating the first two variables. For example, in F2, the success rate is increased by $10.2\times$ and $46.3\times$, respectively, leading to the success rate improved from 0.1\% to 4.6\%. Interestingly, the elimination of the 3rd variable generates limited improvement, from $1.6\times$ to $2.5\times$, which is mainly because most of the non-zero values have been eliminated. Even though this introduces $2\times$ overhead of more measurements, it is acceptable, especially considering the deployability on current NISQ devices.

% The success rate after variable elimination is shown in Figure~\ref{fig:eliminate}~(b). We can see that the success rate is improved by up to $46.3\times$. In F2, the improvement in the first two eliminations can improve the success rate by $34.1 \times$ and $46.3\times$, respectively. The improvement on K2 is only $2.0\times$ since the depth is not reduced heavily. We also observed that after eliminating two variables, more elimination only leading improvement by no more than $2.5\times$. This may be due to the fact that the qubits of commute Hamiltonian after eliminating two variables are quite dispersed, and further elimination has no remarkable effect.

\subsection{Ablation Experiments}
\begin{figure}[t]
    \centering
    \includegraphics[width=.98\linewidth]{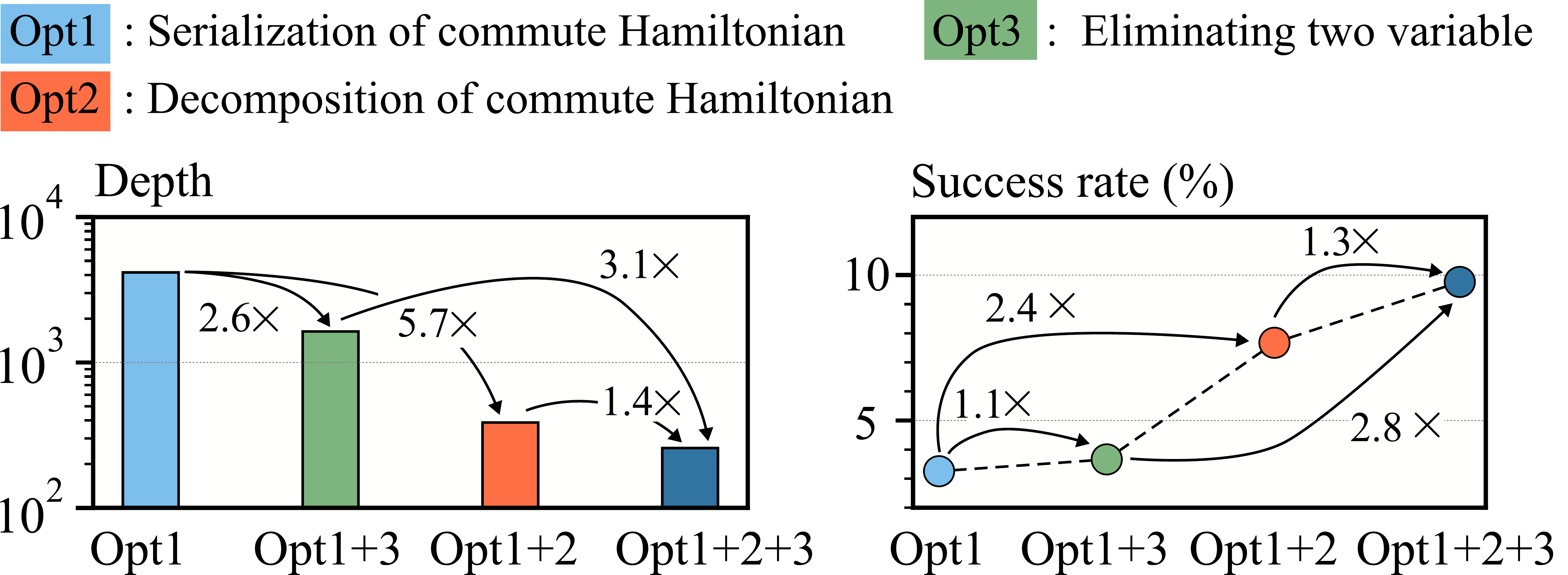}
    \\
    \makebox[0.50\linewidth]{\parbox{0.50\linewidth}{\footnotesize \centering  (a) Ablation study of circuit depth.}}\hfill
    \makebox[0.45\linewidth][c]{\parbox{0.45\linewidth}{\footnotesize (b) Ablation study of success rate.}}
    \\
    \vspace{-0.25cm}
    \caption{Ablation study under IBMQ noise model.}
    \label{fig: ablation}
    \vspace{-0.1cm}
\end{figure}
Figure~\ref{fig: ablation} gives the circuit depth and success rate under different configurations of our proposed optimization techniques. The metrics including depth and success rate are average over the benchmarks and the quantum devices. We have illustrated in Section~\ref{subsec:opteval} that the circuit without serialization of commute Hamiltonian (Opt1) cannot be executed on real-world devices, so we use opt1 in each configuration. Figure~\ref{fig: ablation}~(a) shows that eliminating one variable (Opt3) can reduce circuit depth by $2.6\times$(Opt1+3). Decomposing the local Hamiltonian by phase gate (Opt2) can reduce the circuit depth by $5.7\times$ (Opt1+2) compared to directly decomposing the local Hamiltonian unitary. The circuit depth can be further reduced by $1.4\times$ using variable elimination. The success rate improvement shows a similar trend with the circuit depth as shown in Figure~\ref{fig: ablation}~(b). Decomposition of commute Hamiltonian can improve success rate by $2.4\times$ (Opt1+2). Variable elimination can further increase success rate by $1.3\times$ (Opt1+2+3).

%% file: tab/performance.tex
\begin{table*}[t]
        \centering
        \footnotesize
        \caption{
        The circuit depth, success rate, and in-constraints rate of different QAOA designs under 12 benchmarks.
        }
        \vspace{-0.2cm}
        \label{tab:speed_noise_free}
        \renewcommand\arraystretch{1.05}
        \resizebox{\textwidth}{!}{
        \begin{threeparttable}
        
        \begin{tabular}{ll|llll|llll|llll|llll}
        \toprule 
        \multicolumn{2}{c|}{\multirow{3}{*}{\textbf{Benchmark}}}  & \multicolumn{4}{c|}{\textbf{Success rate (\%)}} & \multicolumn{4}{c|}{\textbf{In-constraints rate (\%)}} & \multicolumn{4}{c|}{\textbf{Approximation ratio gap (ARG)}} & \multicolumn{4}{c}{\textbf{Circuit depth}}\\
        \cline{3-18} 
        \rule{0pt}{13pt} 
        % \cmidrule[0.5pt](rl){6-9} \cmidrule[0.5pt](rl){10-13} \cmidrule[0.5pt](rl){14-17}
         && \makecell[l]{\textbf{Penalty} \\ \textbf{~\cite{verma2022penalty}}} &  \makecell[l]{\textbf{Cyclic} \\ \textbf{~\cite{cyclicdriver}}} &  \makecell[l]{\textbf{HEA} \\ \textbf{~\cite{HEA}}} &  \makecell[l]{\textbf{Choco}\\\textbf{-Q}} 
         & \makecell[l]{\textbf{Penalty}\\\textbf{~\cite{verma2022penalty}}} &  \makecell[l]{\textbf{Cyclic} \\ \textbf{~\cite{cyclicdriver}}} &  \makecell[l]{\textbf{HEA} \\ \textbf{~\cite{HEA}}} & \makecell[l]{\textbf{Choco}\\\textbf{-Q}}
         & \makecell[l]{\textbf{Penalty} \\ \textbf{~\cite{verma2022penalty}}} &  \makecell[l]{\textbf{Cyclic} \\ \textbf{~\cite{cyclicdriver}}} &  \makecell[l]{\textbf{HEA} \\ \textbf{~\cite{HEA}}} & \makecell[l]{\textbf{Choco}\\\textbf{-Q}} 
          & \makecell[l]{\textbf{Penalty} \\ \textbf{~\cite{verma2022penalty}}} &  \makecell[l]{\textbf{Cyclic} \\ \textbf{~\cite{cyclicdriver}}} &  \makecell[l]{\textbf{HEA} \\ \textbf{~\cite{HEA}}} & \makecell[l]{\textbf{Choco}\\\textbf{-Q}} \\ 
        \hline
        
        \multirow{4}{*}{\makecell[l]{\textbf{FLP}\\~\cite{melo2009facilitylocation}}} 
         &\textbf{F1}&3.79&21.4&8.91&\textbf{99.8}                        &22.0 &38.7 &26.4 &\textbf{100.0}                & 39.0  & 15.3  & 44.9 & \textbf{0.16 }   & 56&91&42&43  \\
         &\textbf{F2}&0.14&0.09&\blackxmark &\textbf{54.0}                &0.47 &1.37 &0.10 &\textbf{100.0}                & 107 & 70.9  & 150 & \textbf{0.30}      & 88&99&98&221     \\
         &\textbf{F3} &\blackxmark&0.03&\blackxmark&\textbf{30.0}          &0.04 &0.79 &\blackxmark &\textbf{100.0}         & 145 & 94.6  & 181 & \textbf{0.50}      & 92&134&147&275    \\
         &\textbf{F4}&\blackxmark&\blackxmark&\blackxmark&\textbf{13.3}   &\blackxmark &0.74 &\blackxmark &\textbf{100.0}  & 177 & 97.7  & 257 & \textbf{0.43}      & 108&162&196&421   \\
        \hline
        
        \multirow{4}{*}{\makecell[l]{\textbf{GCP}\\~\cite{graphcoloring}}}
         &\textbf{G1}&0.15&4.74 &0.26&\textbf{69.7}             &0.50 &10.6 &0.36 &\textbf{100.0} & 698 & 217 & 1144 & \textbf{0.06} & 188&230&84&24    \\
         &\textbf{G2} &0.03&0.14&0.03&\textbf{67.1}               &0.07 &0.67 &0.03 &\textbf{100.0} & 747 & 621 & 1727 & \textbf{0.25} & 221&263&105&358  \\
         &\textbf{G3}  &\blackxmark&\blackxmark&\blackxmark&\textbf{17.1} &0.02 &0.27 &\blackxmark &\textbf{100.0}  & 2592 & 479 & 3091 & \textbf{0.47} & 654&708&168&586  \\
         &\textbf{G4}&\blackxmark&\blackxmark&\blackxmark&\textbf{9.50} &\blackxmark &\blackxmark &\blackxmark &\textbf{100.0} & 2366 & 501 & 3834 & \textbf{0.56}& 683&737&196&906   \\
        \hline
        
        \multirow{4}{*}{\makecell[l]{\textbf{KPP}\\~\cite{kpartition}}} 
         &\textbf{K1} &1.66&38.2&1.04&\textbf{86.1}                 &5.18 & 84.8 & 4.46 &\textbf{100.0} & 53.8  & 32.3  & 59.4  &  \textbf{0.14} & 115&150&56&79    \\
         &\textbf{K2}&0.01&14.2&\blackxmark&\textbf{52.6}           &0.05 & 39.7 & 0.04 &\textbf{100.0} & 118 & 37.8  & 130 &  \textbf{0.18} & 202&244&126&324  \\
         &\textbf{K3}&0.01&2.59&\blackxmark&\textbf{21.1}           &0.01 &31.6 &\blackxmark &\textbf{100.0} & 162 & 24.7  &  165 & \textbf{0.24} & 264&306&168&464   \\
         &\textbf{K4}&\blackxmark&0.45&\blackxmark&\textbf{13.3}  &\blackxmark &8.23 &\blackxmark &\textbf{100.0} & 163 & 30.5  & 170 &  \textbf{0.23}& 296&338&189&534    \\
        \hline
        
        \multicolumn{2}{c|}{\textbf{Improv.($\times$)}} & - & -& -& \textbf{$>$235} & - & -& -& \textbf{$>$80.6} & -  & -  & -   & \textbf{658}  & - & -& -&  \textbf{1.00}\\
        \bottomrule
        \end{tabular}
        \begin{tablenotes}
                \item[1] For Penalty-based QAOA~\cite{verma2022penalty}, we integrate it with two open-sourced optimization techniques, \textit{FrozenQubits}~\cite{frozenqubits} and \textit{Red-QAOA}~\cite{redqaoa}. \papername only eliminates one variable. 
                \item[2] The improvement is compared to the cyclic Hamiltonian-based method~\cite{cyclicdriver}. \blackxmark\xspace means that the design fails to find the optimal solution for this case.
        \end{tablenotes}    
        \end{threeparttable}
        }
        \vspace{-0.2cm}
        \end{table*}

%% file: section/8_related_work.tex
\section{Related work}

% \vspace{-0.2cm}
\subsection{Quantum Algorithm for Constrained Binary Optimization}

The first quantum approach to this problem is quantum annealing~\cite{quantumannealing}, which uses the quantum adiabatic theorem to solve unconstrained binary optimization. This approach suffers from long embedding time and long evolution time~\cite{lowquantumannealing}. To overcome this limit, QAOA provides a gate-model version of quantum annealing and uses a variation quantum algorithm to shorten evolution time. Since QAOA can be deployed on current noisy quantum devices, it has been studied and optimized from different aspects, such as compilation~\cite{qaoacompilation}, parameter initialization~\cite{redqaoa}, parameter updating~\cite{streif2020training}, and distributed version~\cite{frozenqubits}. However, quantum annealing and QAOA cannot deal with constraints. Although it can be extended by penalty term~\cite{verma2022penalty}, cyclic Hamiltonian~\cite{cyclicdriver}, they suffer from low success rates and long latency. This paper gives a confident and fast solution by integrating the commute Hamiltonian into QAOA and developing an efficient compilation flow.

Some universal quantum algorithms, alternating from quantum annealing and QAOA, are used to solve this problem with specialized modification. One is the hardware efficient ansatz (HEA)~\cite{HEA,dangwal2023varsaw}. Each variable is encoded by one qubit, and the objective is modified like penalty-QAOA.
Although hardware-efficient, this method cannot always converge into an optimal solution since the circuit structure is not specialized. Another approach named Grover adaptive search uses the Grover algorithm~\cite{long2001grover}, tailored with a selection circuit to exclude the solution outside the constraints~\cite{groveradaptivesearch}. However, the selection circuit is too complex to deploy on hardware. The search process generates too many solutions outside the constraints, requiring enormous iteration to find the target solution. Compared with these works, we use commute Hamiltonian to restrict the search space and multiple optimization strategies to reduce entanglement complexity and circuit depth.

\subsection{Unitary Decomposition for Hamiltonian Simulation}

Hamiltonian simulation simulates the dynamics of a quantum system by decomposing the unitary into a quantum circuit composed of single and two-qubit gates. The most established method is Trotter-Suzuki decomposition~\cite{trotterbasedsimulation}, which uses numerous layers to decompose the total Hamiltonian unitary by repeating multiple local unitary on smaller subsystems. This method can obtain high-precision approximation but with a deep circuit. Other methods, like a linear combination~\cite{hamiltoniansimulationlinear} and quantum signal processing~\cite{hamiltoniansimulationbyqsp}, use a quantum circuit to encode the sparse unitary for simulation. They can approximate Hamiltonian simulation with a shallow circuit but consume extra large ancillary qubit resources. Compared with these works, we provide a precise and efficient method for simulating commute Hamiltonian. We leverage the eigenspace feature of commute Hamiltonian and develop a precise decomposition with linear complexity and linear circuit depth.

%% file: section/9_conclusion.tex
% \vspace{-0.1cm}
\section{Conclusion}
This paper proposes \papername, a fast and hardware-efficient approach for constrained binary optimization problems, which applies commute Hamiltonian that supports arbitrary linear constraints. Then, we develop three optimization techniques to reduce the circuit depth, including Hamiltonian serialization, equivalent decomposition, and variable elimination. These three techniques work together, shrinking the circuit depth from tens of thousands to hundreds. In conclusion, our design greatly improves the success rate and reduces the overall latency of QAOA.

%% file: section/10_acknowledgement.tex
 % \vspace{-0.1cm}
\section{acknowledgements}
This work was supported by National Natural Science Foundation of China (No.62472374) and the National Key Research and Development Program of
China (No. 2023YFF0905200). This work was also funded by Zhejiang Pioneer (Jianbing) Project (No. 2023C01036).

%% file: section/appendix.tex
% \appendix

\section{Artifact Appendix}

%%%%%%%%%%%%%%%%%%%%%%%%%%%%%%%%%%%%%%%%%%%%%%%%%%%%%%%%%%%%%%%%%%%%%
\subsection{Abstract}
In this section, we provide detailed information that will facilitate the artifact evaluation process. The artifact checklist section presents brief information about this artifact and outlines the basic requirements to reproduce the experiment results. Then, we describe the directory tree of our source code and go into more detail about the requirements. Finally, in the experiment workflow section, we provide the usage example and how to reproduce the experiments.
% Our artifact provides the source codes for MorphQV's workflow. We also provide the quantum algorithm benchmarking codes for expressiveness analysis and technique evaluation.

\subsection{Artifact check-list (meta-information)}

% {\em Obligatory. Use just a few informal keywords in all fields applicable to your artifacts
% and remove the rest. This information is needed to find appropriate reviewers and gradually 
% unify artifact meta information in Digital Libraries.}

\begin{itemize}
  \item {\bf Algorithm:} \papername, a formal and universal framework for constrained binary optimization problems. \papername uses the simulation of commute Hamiltonian to embed constraints and achieve a 100\% in-constraints rate to ensure high accuracy. 
  % \item {\bf Algorithm:} \papername, a framework to facilitate confident assertion-based verification in quantum computing. It defines an assertion statement of assume-guarantee assertions and formulates them as a constraint optimization problem with a confidence estimation model to enable rigorous analysis. 
  % \item {\bf Program: } The quantum algorithm benchmarks are included in the artifact.
  \item {\bf Program: } \papername is implemented in Python.
  \item {\bf Data set:} We use problems from 3 domains to evaluate \papername, including facility location, graph coloring, and k-partition. We obtained specific forms of the objective function and constraints from the relevant literature~\cite{melo2009facilitylocation,graphcoloring,kpartition}, and faithfully constructed the relevant dataset in the framework of Choco-Q.
  \item {\bf Run-time environment:}  Ubuntu 20.04 LTS.
  \item {\bf Hardware: }AMD EPYC 9554 64-core CPU, NVIDIA A100 Tensor Core GPU, 64GB Memory, 1TB of storage space.
  \item {\bf Metrics: } In-constraints rate and success rate.
  \item {\bf Results: } The results of experiments are shown in the output of the cells in the .ipynb file and are also written into the .csv file.
  \item {\bf Experiments: } The jupyter notebooks are provided in \texttt{implimentations/} to reproduce the results of the experiments.
  \item {\bf How much disk space is required (approximately)?:} 1TB.
  \item {\bf How much time is needed to prepare workflow (approximately)?:} Less than ten minutes.
  \item {\bf How much time is needed to complete experiments (approximately)?:} Less than two hours. For solving problems with higher problem scales, a GPU simulator will speed up the quantum circuit simulation, making the run time less than one hour. When using the CPU simulator, the quantum circuit simulation needs two hours to be completed.
  \item {\bf Publicly available?:} Yes.
  \item {\bf Code licenses (if publicly available)?:} GNU GPLv3.
  \item {\bf Archived (provide DOI)?:} \url{https://doi.org/10.5281/zenodo.14250941}
\end{itemize}

%%%%%%%%%%%%%%%%%%%%%%%%%%%%%%%%%%%%%%%%%%%%%%%%%%%%%%%%%%%%%%%%%%%%%
\subsection{Description}

\papername is built on Python scripts, which can be executed on Linux systems. Below, we introduce the important files and directories in the artifact. 
% \hl{For a more detailed description, please refer to the README in the GitHub repository.}

\begin{forest}
  for tree={
    font=\ttfamily,
    grow'=0,
    child anchor=west,
    parent anchor=south,
    node distance=2cm, 
    anchor=west,
    calign=first,
    inner xsep=7pt,
    rounded corners,
    font=\small,
    draw=black, fill=cyan!20,
    edge path={
      \noexpand\path [draw, \forestoption{edge}]
      (!u.south west) +(7.5pt,0) |- (.child anchor) pic {folder} \forestoption{edge label};
    },
    before typesetting nodes={
      if n=1
        {insert before={[,phantom]}}
        {}
    },
    fit=band,
    before computing xy={l=15pt},
  }  
[\papername Artifact
  [chocoq/, label={right: {\small core
code of \papername}}
      [model/,  label={right: {\small constrained binary optimization formulation}}
        ]
        [problems/, label={right: {\small dataset of problems}}
        ]
        [solvers/, label={right: {\small code of quantum solvers}}
        ]
  ]
  [implementations/, label={right: {\small jupyter notebooks for experiments}}
    [0\_test.ipynb, label={right: {\small usage of \papername solver}}
    ]
    [1\_table.ipynb, label={right: {\small implimentation of Table~\ref{tab:speed_noise_free}}}
    ]
  ]
  [
  testbed\_cpu.py, label={right: {\small test the installation of CPU version}}
  ]
  [
  testbed\_gpu.py, label={right: {\small test the installation of GPU version}}
  ]
]
\end{forest}
\begin{itemize}
    
    \item {\bf chocoq/.} This sub-directory contains the core code of \papername, which includes the constrained binary optimization model (\texttt{model/}), the problem dataset (\texttt{problems/}), the implementation of \papername solver and other baseline solvers (\texttt{solvers/}). 
    \item {\bf implimentations/.} This sub-directory includes the Jupyter Notebooks to reproduce the experiments in the paper. 
    \item {\bf testbed\_cpu.py. } This Python script gives an example using \papername to define the optimization problem and solve it with \papername with CPU simulator.
     \item {\bf testbed\_gpu.py. } This Python script gives an example using \papername to define the optimization problem and solve it with \papername with a GPU simulator.
\end{itemize}

\textbf{How to access: }
DOI: 10.5281/zenodo.14250941 \\
GitHub: \url{https://github.com/JanusQ/Choco-Q}

\textbf{Hardware dependencies:}
The evaluation in the paper is performed on a server with AMD EPYC 9554 64-core Processor, 1511GB memory, and 32TB storage space. Running the code requires an x86-64 machine with at least 64GB memory and 1TB storage space.

\textbf{Software dependencies:}
The code relies
on Python 3.10 or higher version. We list all required packages in \texttt{environment\_cpu.yml} of the artifact.
\subsection{Dataset}
The data used has been provided in \texttt{chocoq/problems}, which includes the problems in experiments.
\subsection{Installation}

\begin{enumerate}[leftmargin=16pt]
    \item Download the source code from the GitHub repository (\url{https://github.com/JanusQ/Choco-Q.git}) and enter into the \texttt{chocoq/} directory.
    \item Install Anaconda or Miniconda and create a virtual environment from the configuration file.
    \begin{lstlisting}
conda env create -f environment_cpu.yml
    \end{lstlisting}
    \item Activate the virtual environment and install chocoq.
    \begin{lstlisting}
conda activate choco_cpu
pip install .
    \end{lstlisting}
\end{enumerate}

If you want to install the GPU version, please replace the \texttt{cpu} in the previous commands with \texttt{gpu}.

After finishing the above installation, you can run the \texttt{testbed\_cpu.py} to check the installation of the CPU version.
\begin{lstlisting}
python testbed_cpu.py
\end{lstlisting}
To test the installation of the GPU version, run 
\begin{lstlisting}
python testbed_gpu.py
\end{lstlisting}
Please refer to the \texttt{README.md} in the artifact for a more detailed description.
%%%%%%%%%%%%%%%%%%%%%%%%%%%%%%%%%%%%%%%%%%%%%%%%%%%%%%%%%%%%%%%%%%%%%
\subsection{Using \papername}
Users can 1) \textit{define a constrained binary optimization problem} and 2) \textit{solve the problem with Choco-Q solver or other quantum solvers}. We also provide an example code in \texttt{testbed\_cpu.py} to illustrate the usage.
% Users can also adjust the relevant parameters in the provided Python scripts to run other experiments. The adjustable parameters include the range of qubits in benchmarks, the sampling methods, the number of samples, and the optimizer for solving the assertion statements.
%%%%%%%%%%%%%%%%%%%%%%%%%%%%%%%%%%%%%%%%%%%%%%%%%%%%%%%%%%%%%%%%%%%%%
\subsection{Reproducing Experimental Results}
\begin{itemize}
    \item \textbf{Jupyter Notebooks.} The \texttt{1\_table.ipynb} in directory \texttt{implimentations/} provides the Jupyter Notebook to reproduce the experimental results in Table~\ref{tab:speed_noise_free}. You can run it to reproduce the experimental results.
    
\end{itemize}

%%%%%%%%%%%%%%%%%%%%%%%%%%%%%%%%%%%%%%%%%%%%%%%%%%%%%%%%%%%%%%%%%%%%%
% \subsection{Location of experiment results}
% {\em Obligatory}
% \textbf{Results.} 

% The provided scripts for experiments will produce graphs and CSV files in the "examples" directory. Expected results are reported in Section 7, Section 8, and Section 9.
%%%%%%%%%%%%%%%%%%%%%%%%%%%%%%%%%%%%%%%%%%%%%%%%%%%%%%%%%%%%%%%%%%%%%

% \subsection{Notes}
% The number of samples grows exponentially as the number of qubits is asserted, which leads to large memory consumption. A MacOS machine with 8GB memory and the M1 chip can verify the qubit state with up to 9 qubits.
% % Large RAM is needed to verify large quantum states.
% % occupy the RAM storage for inference
% %%%%%%%%%%%%%%%%%%%%%%%%%%%%%%%%%%%%%%%%%%%%%%%%%%%%%%%%%%%%%%%%%%%%%
\subsection{Methodology}
Submission, reviewing, and badging methodology:

\begin{itemize}
  \item \url{https://www.acm.org/publications/policies/artifact-review-badging}
  \item \url{http://cTuning.org/ae/submission-20201122.html}
  \item \url{http://cTuning.org/ae/reviewing-20201122.html}
\end{itemize}